\documentclass[12pt,authoryear]{elsarticle}

\usepackage{amssymb}
\usepackage{amsthm}
\usepackage{mathtools}
\usepackage{easyReview}
\usepackage{graphicx}
\usepackage{lineno}
\usepackage{float}
\usepackage{cite}
\usepackage{amsmath}
\usepackage{physics}
\usepackage{lscape}
\usepackage{graphicx}
\usepackage{booktabs}
\usepackage[authoryear]{natbib}
\newtheorem{definition}{Definition}
\newtheorem{proposition}{Proposition}
\newtheorem{theorem}{Theorem}
\newtheorem{lemma}{Lemma}

\usepackage[linesnumbered,ruled,vlined]{algorithm2e}

\begin{document}

\begin{frontmatter}
	\title{Stochastic Optimal Control of Iron Condor Portfolios for Profitability and Risk Management}
	\author[label2]{Hanyue Huang \corref{cor1}}
 	\author[label1]{Qiguo Sun \corref{cor1}}
        \author[label1]{Xibei Yang}

	\affiliation[label1]{organization={School of Computer, Jiangsu University of Science and Technology},
		city={Zhenjiang},
		postcode={212003},
		state={Jiangsu Province},
		country={China}}
	\affiliation[label2]{organization={Technical University of Munich},
		city={Munich},
		postcode={80333},
		state={Bavaria},
		country={Germany}}
	\cortext[cor1]{indicates equal contribution}

	\begin{abstract}
Previous research on option strategies has primarily focused on their behavior near expiration, with limited attention to the transient value process of the portfolio. In this paper, we formulate Iron Condor portfolio optimization as a stochastic optimal control problem, examining the impact of the control process \( u(k_i, \tau) \) on the portfolio's potential profitability and risk. By assuming the underlying price process as a bounded martingale within $[K_1, K_2]$, we prove that the portfolio with a strike structure of $k_1 < k_2 = K_2 < S_t < k_3 = K_3 < k_4$ has a submartingale value process, which results in the optimal stopping time aligning with the expiration date $\tau = T$.
Moreover, we construct a data generator based on the Rough Heston model to investigate general scenarios through simulation. The results show that asymmetric, left-biased Iron Condor portfolios with $\tau = T$ are optimal in SPX markets, balancing profitability and risk management. Deep out-of-the-money strategies improve profitability and success rates at the cost of introducing extreme losses, which can be alleviated by using an optimal stopping strategy. Except for the left-biased portfolios $\tau$ generally falls within the range of [50\%,75\%] of total duration. In addition, we validate these findings through case studies on the actual SPX market, covering bullish, sideways, and bearish market conditions.
	\end{abstract}
	\begin{keyword}

		Iron Condor \sep Fractional Brownian Motion  \sep Rough Heston \sep Optimal Stopping Time

	\end{keyword}

\end{frontmatter}


\section{Introduction}

Option portfolio optimization can be formulated as a stochastic optimal control problem, where the portfolio value process indicated by $X_t$ is a state process, which can be partly controlled by some control process $u$.
For a given initial point $x_0 \in \mathbb{R}^n$, we consider $X_t$ as a controlled stochastic differential equations:

\begin{equation}
        \begin{cases}
	dX_t = \mu (t,X_t,u_t)dt + \sigma(t,X_t,u_t)dW_t, \\
        X_0 = x_0
        \end{cases}
	\label{eq:value process}
\end{equation}
where $\mu$ and $\sigma$ are some given functions with
\begin{equation}
	\begin{cases}
		\mu: \mathbb{R}^+ \times  \mathbb{R}^n \times  \mathbb{R}^k \rightarrow  \mathbb{R}^n, \\
		\sigma: \mathbb{R}^+ \times  \mathbb{R}^n \times  \mathbb{R}^k \rightarrow  \mathbb{R}^{n\times d}
	\end{cases}		
	\label{eq:mu and sigma}
\end{equation}
where $W$ is a d-dimensional Wiener process. The control process $u$ is defined to be adapted to the $X_t$, denoted by $\mathbf{u}(t,X_t)$.
Moreover, it is also an admissible function in the admissible control law classes $U$ for all $t \in \mathbb{R}^+ $ and  $x \in \mathbb{R}^n $.

In real applications, the state process $X_t$ must reside within a fixed domain. Consider a general situation where all options in the portfolio share the same expiration date, $T$. Moreover, let us assume a group of option strategies where the maximum potential loss and maximum potential gain are predetermined at the initial state $t=0$. This setup naturally forms a domain for such a class of option portfolios, denoted as $D \subseteq [0,T] \times \mathbb{R}^n$. 
When the state process reaches the boundary $\partial D$, the investment terminates. Thus, we define the stopping time as follows:
\begin{equation}
    \tau = \inf \{ t \geq 0 \mid (t, X_t) \in \partial D \} \wedge T,
\label{eq:stopping_time}
\end{equation}
where $x \wedge y = \min(x, y)$.

Furthermore, we define the instantaneous utility function as $F(t, x, u)$ and the terminal utility (or bequest) as $\Phi(\tau, X_\tau)$, where $\Phi: \partial D \to \mathbb{R}$. The optimization problem is then to maximize the following expected value:
\begin{equation}
    \mathbb{E} \left[ \int_0^{\tau} F(s, X_s, u_s) \, ds + \Phi(\tau, X_\tau) \right].
\label{eq:optimization_problem}
\end{equation}

A well-known option strategy belonging to the above class is the Iron Condor strategy \citet{cohen2005bible,woodard2011iron}. Without loss of generality, $X$ only contains a underlying process $S_t$ and a bank account $B_t$ process as follows:
\begin{equation}
	\begin{cases}
		dS_t = \mu S_t dt + \sigma S_t dW_t, \\ 
            dB_t = rB_t dt
	\end{cases}		
	\label{eq:simplified process}
\end{equation}
where $r$ is the short rate. From general portfolio theory $X$ process can be expressed as,
\begin{equation}
    dX_t =  X_t((1-\omega)\frac{dB_t}{B_t}+ \omega \frac{dS_t}{S_t})
\label{eq:portfolio theory}
\end{equation}
where $\omega$ is the weight of the underlying asset.

If we do not consider optimizing consumption and set $F=0$,  the utility function depends solely on the terminal wealth, which is the usual case for Iron Condor portfolio, then the control process $u \in U$, where $U$ is the family of Iron Condor portfolios corresponding to some underlying, is determined by the structure of the four options and time t. Specifically, Iron Condor combines a bullish put spread and a bearish call spread with option strikes satisfying $k_1 < k_2 \leq k_3 < k_4$. A concrete example of the terminal Profit and Loss (P\&L) profile is illustrated in Figure \ref{Fig:Iron condor Expiration}. 

Most previous research on the Iron Condor focused on portfolio behavior at or near the expiration date \citet{de2023simple, dziawgo2020iron}. For instance, \citet{de2023simple} conducted a comparative study on the success rates of put spreads, call spreads, and Iron Condor portfolios near expiration, using the SPX dataset. Their findings indicated that put and call spreads generally exhibit higher success rates compared to Iron Condor portfolios. In addition, the success rates of all strategies decrease as the time to expiration increases.
In another study, \citet{dziawgo2020iron} analyzed the risk measures of Iron Condor portfolios through the lens of option Greeks. Their results revealed that all risk metrics fluctuate significantly over time.

In another study, \citet{dziawgo2020iron} analyzed the risk measures of Iron Condor portfolios based on option Greeks. Their results revealed that all risk metrics fluctuate significantly with the increase of time horizon. Despite these contributions, previous research did not dive into the transient behavior of the Iron Condor portfolio. A comprehensive investigation of the potential profits and risks in this aspect, either through theoretical proofs or simulation-based analyses, is required.

\begin{figure}[H]
	\centering
	\includegraphics[trim=0cm 0cm 0cm 0cm,clip=true,width=12cm]{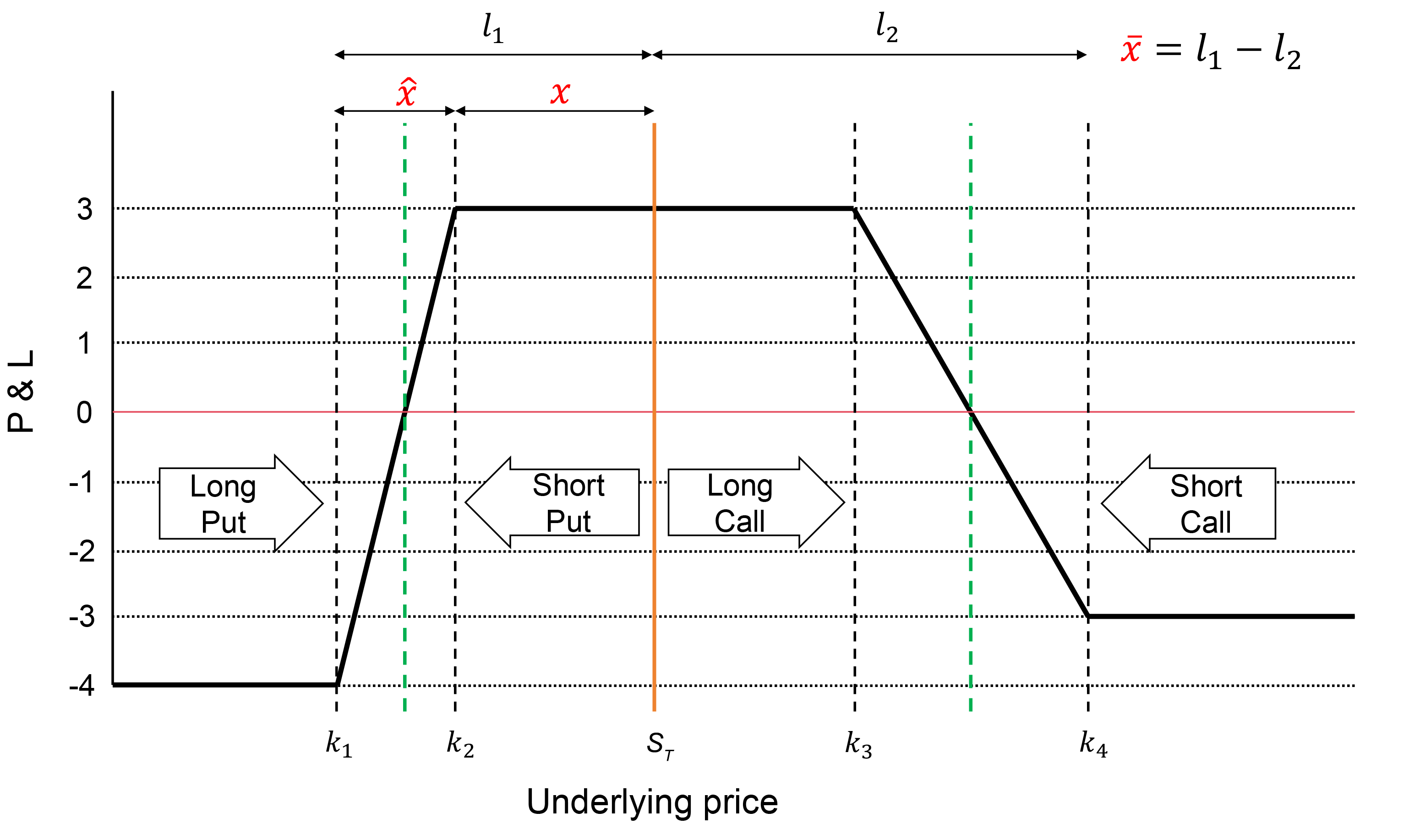}
	\vspace*{-3mm}
	\caption{
		Schematic diagram of $P\&L$ of an Iron Condor portfolio at expiration $T$, where $S_T$ represents the current underlying price, and $k_1$, $k_2$, $k_3$, and $k_4$ correspond to the strike prices of the long put, short put, long call, and short call options, respectively. The two green dashed lines indicate the breakeven prices of the put and call spreads. The control process $u$ is paramtrized by $x$ (moneyness), $\hat{x}$ (span) and $\bar{x}$ (asymmetry degree). }		
	\label{Fig:Iron condor Expiration}
\end{figure}

Financial stochastic process simulation models have undergone extensive development, particularly in addressing volatility-related challenges.
The classical Black-Scholes model assumes a constant volatility surface, which is an unrealistic simplification. To address this limitation, \citet{dupire1994pricing} developed local volatility models that simulate volatility as a deterministic function of the underlying price and time, effectively capturing time-inhomogeneity. Meanwhile, \citet{heston1993closed} proposed the renowned Heston model, which describes volatility using stochastic differential equations driven by Brownian motion, incorporating additional dynamics such as mean reversion. This results in a semi-martingale process. However, these classical models often fail to capture the full complexity of observed option prices.

\citet{gatheral2014arbitrage} identified that realized log-volatility behaves like a fractional Brownian motion (fBm) with a Hurst exponent $H$ around 0.1. This insight spurred the development of rough volatility models. The Rough Fractional Stochastic Volatility (RFSV) model, introduced by \citet{bayer2016pricing}, demonstrated remarkable consistency with observed SPX volatility surfaces. To integrate the rough volatility dynamics with the analytical tractability of the classical Heston model, the Rough Heston model was developed. This model replaces the standard Brownian motion with fBm characterized by $H < 0.5$.

Despite their successes, rough stochastic models are computationally intensive. \citet{bennedsen2017hybrid} developed a hybrid scheme to enhance simulation efficiency, applying it to Monte Carlo pricing in the rough Bergomi model,
 achieving faster and more accurate fits for implied volatility smiles. Markovian approximation techniques, such as those in \citet{abi2019lifting}, 
 further speed up rough volatility models by approximating the power-law kernel using a system of exponential kernels.
\citet{fadugba2020homotopy} use homotopy analysis method to price European call option with time-fractional BS equation. 
\citet{wang2022practical} employ finite difference method to study the multi-dimensional fractional
Balck-Scholes model under three underlying assets.
Recently, \citet{wong2024simulation} introduced a fast algorithm for simulating fBm-driven processes, achieving a tenfold speed increase compared to traditional rough Heston simulations.

In this work, we first provide a theoretical proof of the optimal stopping time for an Iron Condor portfolio with specific structures where the underlying prices follow a bounded martingale. Next, we utilize the Rough Heston model and its fast simulation algorithms to design a data generator to investigate the potential profits and risks associated with Iron Condor portfolios under more general conditions.

Specifically, let
$(\omega,\mathcal{F},P,\mathbb{F})$ denote a filtered probability space, where the filtration $\mathbb{F}$ satisfies the usual conditions. 
Our objective is to determine the optimal control process  $u(k_i,\tau) \in \mathbb{F}_t,i\in \{1,2,3,4\}$  (or its parameterized form $u(x,\hat{x},\bar{x},\tau)) $ for all $\tau < t$.  
The parameterized study is based on simulation methods, where the three key portfolio structure parameters used in the simulations are moneyness ($x$), strike span ($\hat{x}$), and asymmetry degree ($\bar{x}$).

The remainder of this paper is organized as follows: Section 2 details the methodology and metrics used in this study. Section 3 presents the theoretical proof of the optimal control strategy for a bounded martingale process. Section 4 focuses on the symmetric Iron Condor portfolio simulation and analysis. Section 5 investigates the asymmetric Iron Condor portfolio. Section 6 validates the findings from the simulation on actual SPX datasets across bullish, sideways, and bearish markets. Finally, Section 7 concludes the paper and suggests potential directions for future research.

\section{Methodology}
\subsection{Data generator implementation}

The data generator is utilized for simulation purposes, generating underlying asset paths and performing Monte Carlo-based option pricing for the simulation of Iron Condor portfolios.
To design the data generator, we must add more structural assumptions for $S_t$, so the univariate Rough Heston framework \citet{el2019roughening} is employed.  
The volatility term in Rough Heston evolves as a correlated fractional Brownian motion (fBm) with Hurst parameter $H < 0.5$, reflecting its self-similar Gaussian process nature. The $S_t$ process is defined as follows:

\begin{definition}
The underlying asset prices have P-dynamics defined as
\begin{equation}
    dS_t = \mu S_t \, dt + \sqrt{V_t} S_t \, dW_t^1,
\end{equation}
with 
\begin{equation}
    V_u = V_t + \frac{\kappa}{\Gamma\left(H + \frac{1}{2}\right)} \int_t^u \frac{\theta_t(s) - V_s}{(u - s)^{\frac{1}{2} - H}} \, ds 
    + \frac{\nu}{\Gamma\left(H + \frac{1}{2}\right)} \int_t^u \frac{\sqrt{V_s}}{(u - s)^{\frac{1}{2} - H}} \, dW_s^2,
\end{equation}
where:
$\mu$ is he drift term of the asset price, representing the expected return rate,
$V_t$ is the stochastic volatility of the underlying asset at time $t$,
$W_t^1$ is a standard Brownian motion driving the asset price process,
$W_s^2$ is a standard Brownian motion independent of $W_t^1$ that driving the volatility process,
$\kappa$ is the mean reversion rate, $\theta_t(s)$ is the long-term, $F_t$ measurable function indicating the mean level of the variance process, 
$\Gamma(H + \frac{1}{2})$ is the Gamma function evaluated at $H + \frac{1}{2}$, scaling the rough volatility dynamics,
$\nu$ is the volatility of volatility parameter, influencing the intensity of the volatility process.
The stationary increments of fractional Gaussian noise (fGn) has an autocovariance function defined as,
\begin{equation}
    \rho_H(k) = \frac{1}{2}(|k+1|^{2H} + |k-1|^{2H} - 2|k|^{2H}), \quad k \in \mathbb{R}^+. 
\end{equation}
\end{definition}

Since the subsequent experiments are conducted on the SPX option chain, the params of the Rough Heston model adopts $r = 0$, $V_0 = 0.0392$, $\kappa = 0.1$, $\theta = 0.3156$, $\nu = 0.0331$, and $\rho = -0.681$, which is calibrated by \citet{ma2022fast}.
Moreover, the fast algorithm proposed by \citet{wong2024simulation} is employed to perform Monte Carlo-based option pricing at each time step prior to maturity. The risk-neutral measure for the fractional Brownian motion (fBm) process is adopted using an fBm-specific version of Girsanov's theorem \citet{hu2003fractional}.

We consider relative strikes for call and put options within the range $k \in [0.8, 1.2]$ with increments of $0.2$, over the time horizon $t \in [0, T]$. The Monte Carlo simulation uses 10,000 trajectories, and each procedure is repeated 30 times to ensure robust results.

\subsection{Iron Condor Portfolio}

An Iron Condor strategy is a combination of a bullish put spread and a bearish call spread, the investors achieve maximum profit if $S_T$ remains between $k_2$ and $k_3$. An Iron Condor portfolio adopt an adapt control process $u(k_i,\tau)$ or the the parametrized form $u(x,\hat{x},\bar(x),\tau)$, which is  defined as follows (see also Figure 1 but use $S_0$ rather than $S_T$):

$x$ measures the relative strike position to the current underlying price, defined as:
\begin{equation}
    x = \frac{|k_2 - S_0|}{S_0},
\end{equation}
where $S_0$ represents the underlying price, and $k_2$ is the strike of the short put.

$\hat{x}$ measures the distance between $k_1$ and $k_2$ (or $k_4$ and $k_3$ due to symmetry), defined as:
\begin{equation}
    \hat{x} = (k_2 - k_1).
\end{equation}

The asymmetry, $\bar{x}$, measures the imbalance of moneyness between the bullish put spread and the bearish call spread, defined as:
\begin{equation}
    \bar{x} = (S_0 - k_1) - (k_4 - S_0),
\end{equation}
where $k_1$ and $k_4$ are the strike prices of the first and fourth options, respectively

\begin{definition}
    An \textbf{Iron Condor} portfolio \( P_t(u(k_i,\tau)) \) is a function of the control process $u$ and stopping time $\tau$, such that 
     \( P_t(u(k_i,\tau)) \) has the following dynamics: 
\begin{equation}
	\begin{cases}
		P_{u,\tau} = (V_P(k_2, S_t,\sigma,t) - V_P(k_1,S_t,\sigma, t)) + (V_C(k_3,S_t,\sigma, t) - V_C(k_4,S_t,\sigma, t)),  \\
		P_{0,k} = 0
	\end{cases}		
	\label{eq:normalized portfolio}
\end{equation}
subject to the constraints
\begin{equation}
	\begin{cases}
		k_1 < k_2 \leq k_3 < k_4; \\
		\tau \leq t 
	\end{cases}		
	\label{eq:constrain}
\end{equation}
Where $V_p$ and $V_c$ are Monte Carlo option pricing functions for put and call options.
\end{definition}

\subsection{Datasets Partition}

Due to the complex nonlinear structure of the payoff of Iron Condor portfolios, deriving the optimal stopping strategy under general conditions is challenging, so we conduct simulations using the data generator. 

The overall generated dataset is denoted by \( D_{n,t,f} \in \mathbb{R}^{N \times T \times F} \), where:

\begin{itemize}
    \item \( N \) is the number of underlying prices trajectories.
    \item \( T \) is the total time steps to maturity.
    \item \( F \) is the number of portfolios.
\end{itemize}

We note that $D_{n,t,0}$ are all the underlying price processes, and $D_{n,t,i}, i \in [1,F]$ are all the normalized value processes of Iron Condor portfolios under different control $u(k_i,T)$ for the n-th underlying price process. 

Therefore, dataset partition is based on the 0-th dimension of F for N underlying price process, defined as follows:

\begin{definition}
The datasets of bullish market $D_r$, sideway market $D_M$, and bearish market $D_l$ are the partition of $D$ with the following rule: 
\begin{equation}
    \begin{cases} 
        D_r := D_{n,t,f}|\frac{D_{n,T,0}}{D_{n,0,0}} \in [1.1, +\infty],\\
        D_M := D_{n,t,f}|\frac{D_{n,T,0}}{D_{n,0,0}} \in [0.9, 1.1], \\
        D_l := D_{n,t,f}|\frac{D_{n,T,0}}{D_{n,0,0}} \in [-\infty, 0.9]
    \end{cases}
    \label{eq:dataset partition}
\end{equation} 
\end{definition}

\subsection{Variable and indicator design}

Moreover, To simplify subsequent analysis, we only consider the options portfolio payoff process, and set weight of $B_t$ process in Equation \ref{eq:portfolio theory} to 0.
Furthermore, we normalize \( P_t(u(k_i,\tau)) \) to get the potential profit under control $u$ denoted by $\phi_{t,u}$, defined as:
\begin{equation}
        \phi_{t,u} = \mathbb{E}_{\omega} [ \frac{P_0(\omega,t,u)-P_t(\omega,t,u)}{P_0(\omega,t,u)} |D ], \quad  t \in [0,T], \quad u \in U
	\label{eq:norm p}
\end{equation}

Consequently, the potential profit under control $u$  at the expiration date is denoted by $\phi_{T,u}$, and the potential profit at optimal stopping time is denoted by $\phi_{\tau,u}$. We note that the optimal stopping time is determined on dataset $D$ and then used as a constant in calculating other metrics.  
Formally, 
\begin{equation}
	\begin{cases}
	\phi_{\tau,u}=  \mathbf{max} \quad \mathbb{E}[\phi_{t,u} |D], \quad  t \in [0,T], \quad u \in U \\
	\tau(u) = \mathbf{argmax} \quad \mathbb{E} [\phi_{t,u} |D], \quad  t \in [0,T], \quad u \in U \\
	\end{cases}
	\label{eq:phi_tau}
\end{equation}

Next, we define the success rates at time $T$ and optimal stopping time $\tau$ as:

\begin{equation}
	\begin{cases}
	\theta_{T,u}=  \mathbb{E}_{\omega} [\mathbb{I}_{ \phi_{T,u}>0})|D], \quad  t =T, \quad u \in U \\
	\theta_{\tau,u}=  \mathbb{E}_{\omega} [\mathbb{I}_{ \phi_{\tau,u}>0})|D], \quad  t = \tau, \quad u \in U \\
	\end{cases}
	\label{eq:theta}
\end{equation}

The two potential profit metrics for sideways market dataset $D_M$ are denoted by $\phi_{T,M}$ and $\phi_{\tau,M}$, in which we use the same $\tau$ values in $\phi_{\tau}$ but condition on dataset $D_M$.

Finally, the risk $\eta_{u}$ results from significant bullish or bearish markets are measured based on $D_r$ and $D_l$ data, which is defined as:

\begin{equation}
	\begin{cases}
	\eta_{T,u}|D_r = \mathbb{E}_{\omega} [ \phi_{T,u} | D_r] , \quad  t =T, \quad u \in U \\
	\eta_{T,u}|D_l = \mathbb{E}_{\omega} [ \phi_{T,u} | D_l] , \quad  t =T, \quad u \in U  \\
	\end{cases}
	\label{eq:eta}
\end{equation}

\section{Optimal Control for a Bounded Martingale Process}

This section analyzes a simplified case where \( S_t \) is a bounded martingale under the risk-neutral measure \( \mathbb{Q} \), and \( K_2 \leq S_t \leq K_3 \) for all \( t \in [0, T] \).

We begin with the following lemma:

\begin{lemma}
\label{lem:low_volatility}
Let \( S_t \) be a martingale, and assume that \( S_t \) remains within the profitable region \([K_2, K_3]\) for all \( t \in [0, T] \). Then the value process of Iron Condor portfolio \( P_{u, t} \) is a submartingale. 
\end{lemma}

\begin{proof}

The rate of time decay (\( \Theta \)) for short and long options is given by:
\begin{equation}
\begin{cases}
\Theta_{\text{long}} = - \frac{\partial V_C}{\partial t} \quad \text{or} \quad - \frac{\partial V_P}{\partial t},\\
\Theta_{\text{short}} =  \frac{\partial V_C}{\partial t} \quad \text{or} \quad \frac{\partial V_P}{\partial t},
\end{cases}
\label{eq:cal_theta}
\end{equation}
where \( V_C \) and \( V_P \) denote the call and put option prices, respectively.

According to the no-arbitrage theory, given the relation $K_1 < K_2 \leq S_t \leq K_3 < 4$, we have 
\begin{equation}
\begin{cases}
V_{C,t}(K_1) \leq V_{C,t}(K_2),\\
V_{P,t}(K_4) \leq V_{P,t}(K_3)
\end{cases}
\label{eq:ineq_price}
\end{equation}

Applying dynamic programming lead to the following relations,

\begin{equation}
\begin{cases}
|\Theta(K_1,t)| \leq |\Theta(K_2,t)|, \\
|\Theta(K_4,t)| \leq  |\Theta(K_3,t)| 
\end{cases}
\label{eq:ineq_theta}
\end{equation}

We have thus obtained the following inequality:
\begin{equation}
    \Theta(K_1,t) + \Theta(K_2,t) + \Theta(K_3,t) + \Theta(K_4,t) \geq 0, \forall t \in [0,T]
\end{equation}

Since \( S_t \) is within the profitable region \([K_2, K_3]\), the intrinsic values of all options remain \( 0 \), and changes in portfolio value are driven by extrinsic time decay. So we obtain the inequality

\[
\mathbb{E}[P_{u, t+1} \mid \mathcal{F}_t] \geq P_{u,t}, \quad \forall t \in [0, T].
\]
Thus, \( P_{u,t} \) is a submartingale. This concludes the proof of Lemma~\ref{lem:low_volatility}.
\end{proof}

We now recall a foundational proposition in the theory of optimal stopping, which serves as a basis for analyzing American-style options.

\begin{proposition}
\label{prop:optimal_stopping}
\begin{enumerate}
    \item If the discounted payoff process is a submartingale, then late stopping is optimal, i.e., \( \hat{\tau} = T \).
    \item If the discounted payoff process is a supermartingale, then it is optimal to stop immediately, i.e., \( \hat{\tau} = 0 \).
    \item If the discounted payoff process is a martingale, then all stopping times \( \tau \) with \( 0 \leq \tau \leq T \) are optimal.
\end{enumerate}
\end{proposition}

Using Lemma~\ref{lem:low_volatility} and Proposition~\ref{prop:optimal_stopping}, we establish the following theorem:

\begin{theorem}  
\label{the:hold_to_end}  
Let \( S_t \) be a martingale bounded by \( K_2 \) and \( K_3 \) for all \( t \in [0, T] \). Then, for any Iron Condor portfolio with the strike structure \( k_1 \leq k_2 = K_2 < k_3 = K_3 < k_4 \), the optimal stopping time is \( \tau = T \).  
\end{theorem}

\begin{proof}
Define the Iron Condor portfolio values \( P_{u,T} \) as the discounted payoff process adapted to \( S_t \), where
\begin{equation}
\mathbb{E}[S_{t+1} \mid \mathcal{F}_t] = S_t,
\label{eq:st_martingale}
\end{equation}
subject to \( K_2 \leq S_t \leq K_3 \) for all \( t \in [0, T] \).

From Lemma~\ref{lem:low_volatility}, \( P_{u,T} \) is a submartingale. By Proposition~\ref{prop:optimal_stopping}, the optimal stopping time is \( \tau = T \). This concludes the proof of Theorem~\ref{the:hold_to_end}.
\end{proof}

\section{Simulation Research on Symmetric Iron Condor}

\subsection{Influence of Moneyness}
In this section, Moneyness $x$ is the sole variable determining the control process $u$. Figures \ref{Fig:ms_1}, \ref{Fig:ms_4}, and \ref{Fig:ms_7} depict the distributions of potential profits $\phi_{t,u}(\omega)$ and risks $\eta_{t,u}(\omega)$ for $x$ values of 0.12, 0.06, and 0.00, respectively.
The red line in each figure is the expectation process, denoted by $\phi_{t,u}$ or  $\eta_{t,u}$.

The options in Figure \ref{Fig:ms_1} are deep Out-Of-The-Money (OTM), resulting in the $\phi_{t,u}(\omega)$ distribution spanning a wide range of values from [-4,1] and $\phi_{t,u} \approx 0$ for all $t \in [0,T]$ (Figure \ref{Fig:ms_1} (a)).  
In a sideways market (Figure \ref{Fig:ms_1}(b)), $\phi_{t,u}$ exhibits a convex increase. However, the risk associated with this portfolio is significant, as illustrated in Figures \ref{Fig:ms_1}(c) and \ref{Fig:ms_1}(d), where $\eta_{t,u}$ shows a concave decrease as time approaches the option expiration date.

\begin{figure}[H]
	\centering
	\includegraphics[trim=0cm 0cm 0cm 0cm,clip=true,width=10cm]{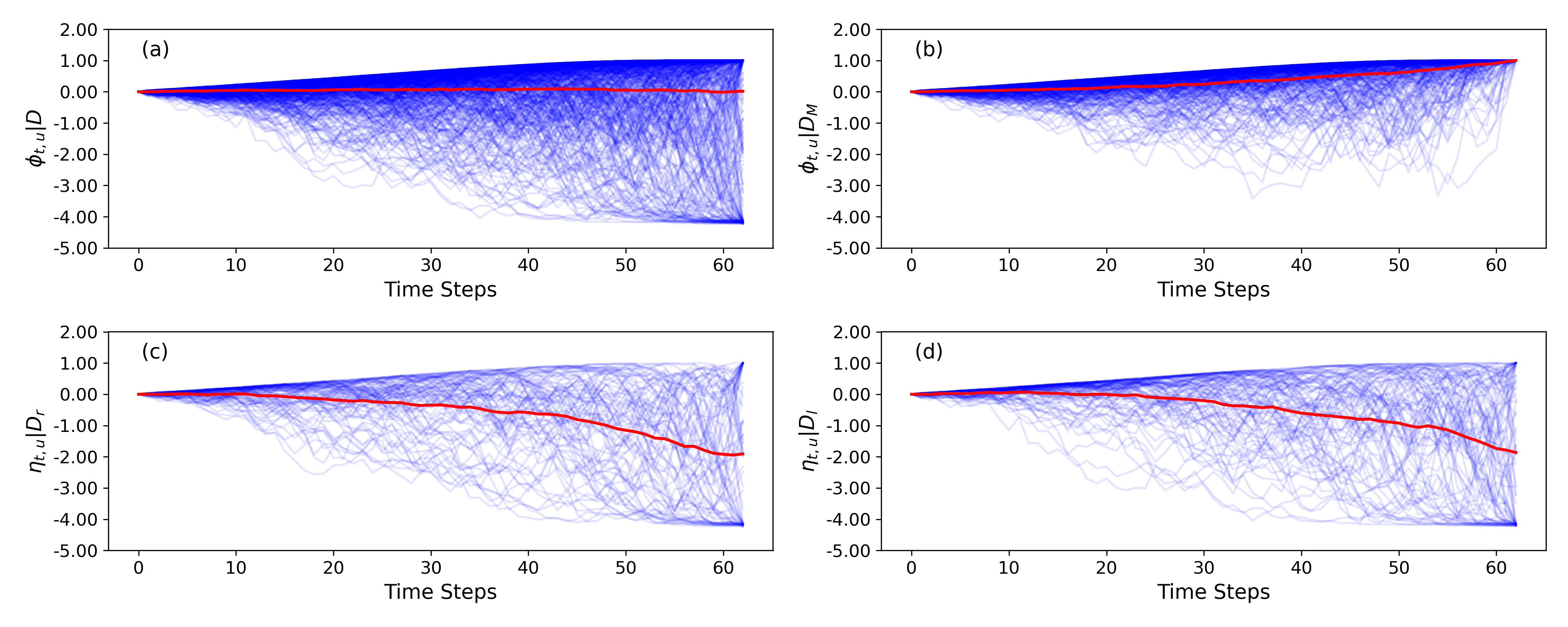}
	\vspace*{-3mm}
	\caption{Distributions of (a) $\phi_{t,u}(\omega)|D$; (b) $\phi_{t,u}(\omega)|D_M$; (c) $\eta_{t,u}(\omega)|D_r$ and (d) $\eta_{t,u}(\omega)|D_l$ for the deep OTM portfolio with $x=0.16,  k=[0.84, 0.88, 1.12, 1.16]$. The red lines represent the expectations.}
	\label{Fig:ms_1}
\end{figure}

In contrast, the options in Figure \ref{Fig:ms_4} are slightly OTM, meaning their strike prices are closer to the underlying price compared to the previous case. In Figure \ref{Fig:ms_4}(a), the distribution of $\phi_{t,u}(\omega)$ becomes denser within the range of [-1.2,1]. In Figures \ref{Fig:ms_4}(c) and \ref{Fig:ms_4}(d), both $\eta_{t,u}|D_r$ and $\eta_{t,u}|D_l$ exhibit a concave decrease over time, but the maximum loss is significantly reduced compared to that shown in Figures \ref{Fig:ms_1}(c) and \ref{Fig:ms_1}(d).

\begin{figure}[H]
	\centering
	\includegraphics[trim=0cm 0cm 0cm 0cm,clip=true,width=10cm]{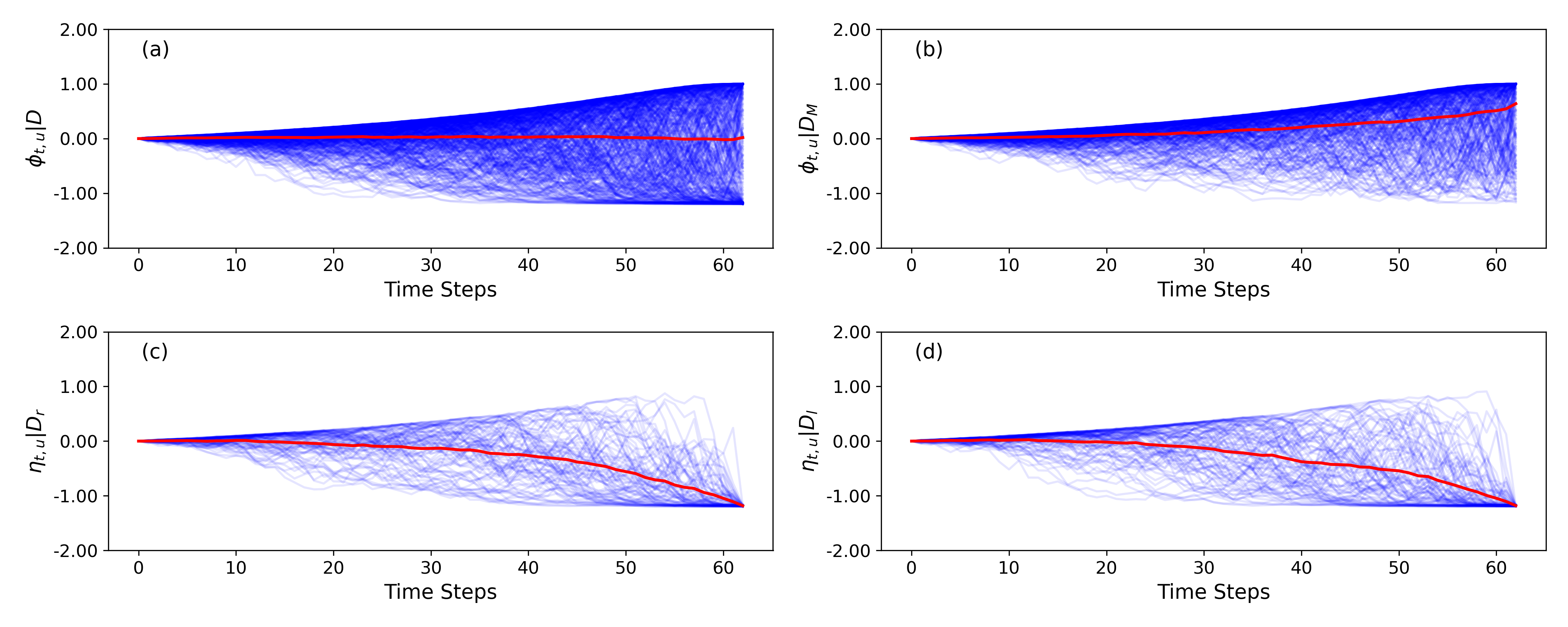}
	\vspace*{-3mm}
	\caption{Distributions of (a) $\phi_{t,u}(\omega)|D$; (b) $\phi_{t,u}(\omega)|D_M$; (c) $\eta_{t,u}(\omega)|D_r$ and (d) $\eta_{t,u}(\omega)|D_l$ for the slightly OTM portfolio with $x=0.10, k=[0.9, 0.94, 1.06, 1.1]$. The red lines represent the expectations.}
	\label{Fig:ms_4}
\end{figure}

Figure \ref{Fig:ms_7} presents the at-the-money (ATM) case, where $x=0$ and $k_2 = S_0 = k_3$, transforming the Iron Condor portfolio into an Iron Butterfly structure.  
The distribution of $\phi_{t,u}(\omega)$ in Figure \ref{Fig:ms_7}(a) is particularly appealing to investors, as it retains the maximum profit potential (albeit with a low probability) while capping the maximum loss at -0.2. However, Figure \ref{Fig:ms_7} shows that even under profitable market conditions, the $\phi_{t,u}|D_M$ is significantly lower compared to the previous cases where $x=0.12$ and $x=0.06$.

\begin{figure}[H]
	\centering
	\includegraphics[trim=0cm 0cm 0cm 0cm,clip=true,width=10cm]{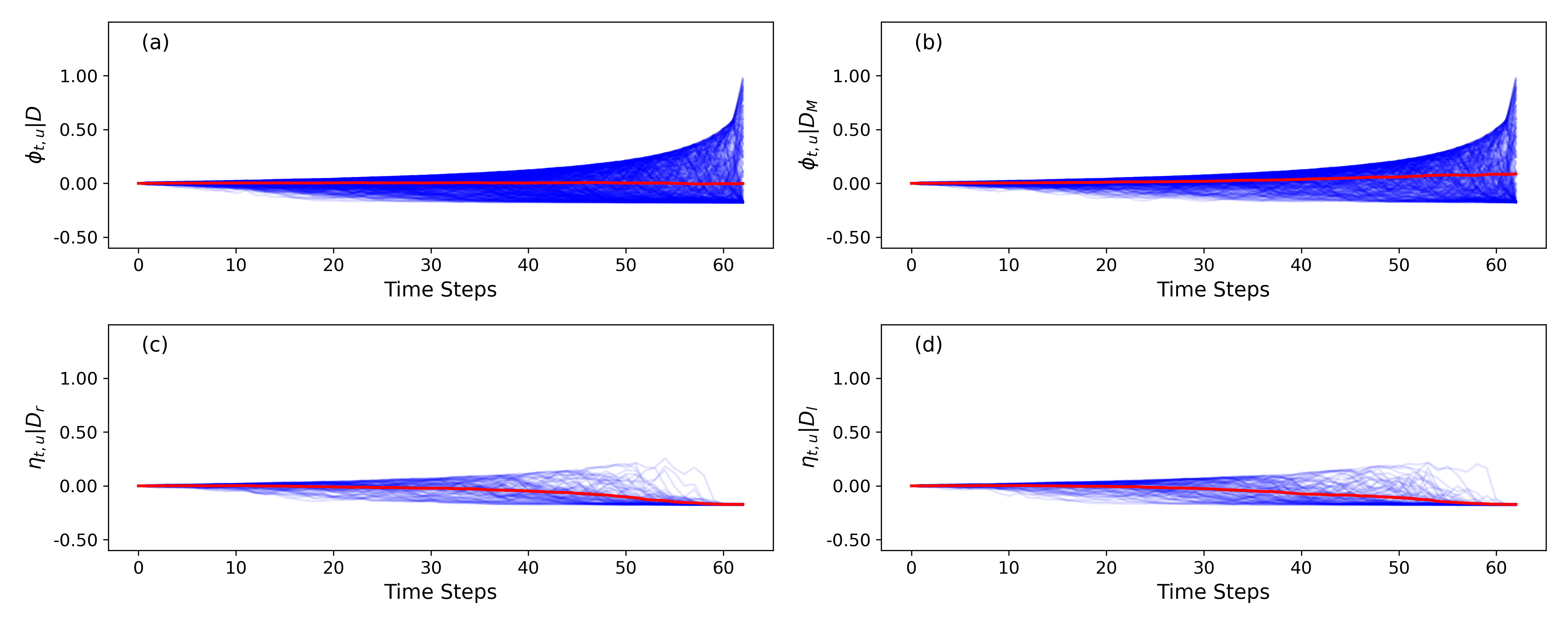}
	\vspace*{-3mm}
	\caption{Distributions of (a) $\phi_{t,u}(\omega)|D$; (b) $\phi_{t,u}(\omega)|D_M$; (c) $\eta_{t,u}(\omega)|D_r$ and (d) $\eta_{t,u}(\omega)|D_l$  for the At-The-Money (ATM) portfolio with $x=0.00, k=[0.96, 1.0, 1.0, 1.04]$. The red lines represent the expectations.}
	\label{Fig:ms_7}
\end{figure}

Figure \ref{Fig:ms} summarizes the influence of $x$ on portfolio performance.  
In Figure \ref{Fig:ms}(a), $\phi_{\tau,u}|D$ increases exponentially with $x$. The growing deviation between $\phi_{\tau,u}|D$ and $\phi_{T,u}|D$ as $x$ increases highlights the importance of employing an optimal stopping strategy for deep OTM portfolios.  

In Figure \ref{Fig:ms}(b), $\phi_{\tau,u}|D_M$ underperforms $\phi_{T,u}|D_M$, suggesting that applying an optimal stopping strategy in sideways market conditions reduces profitability. This observation is consistent with the theoretical proof presented in Section 3.  

In Figure \ref{Fig:ms}(c), both $\theta_{T,u}|D$ and $\theta_{\tau,u}|D$ increase with $x$. This is because increasing $x$ widens the span between $k_2$ and $k_3$, thereby improving the success ratio. However, an interesting crossover is observed between $\theta_{T,u}|D$ and $\theta_{\tau,u}|D$, indicating that for deep OTM options, holding the portfolio until expiration outperforms early stopping at the optimal time $\tau$.  

In Figure \ref{Fig:ms}(d), $\theta_{T,u}$ and $\theta_{\tau,u}$ exhibit similar performance under the $D_r$ and $D_l$ datasets. Notably, $\theta_{\tau,u}$ consistently outperforms $\theta_{T,u}$, indicating that applying an optimal stopping strategy generally reduces risks in volatile market conditions. For deep OTM options with $x > 0.1$, employing an optimal stopping strategy can significantly truncate risk within the range of -1.0.  

\begin{figure}[H]
	\centering
	\includegraphics[trim=0cm 0cm 0cm 0cm,clip=true,width=10cm]{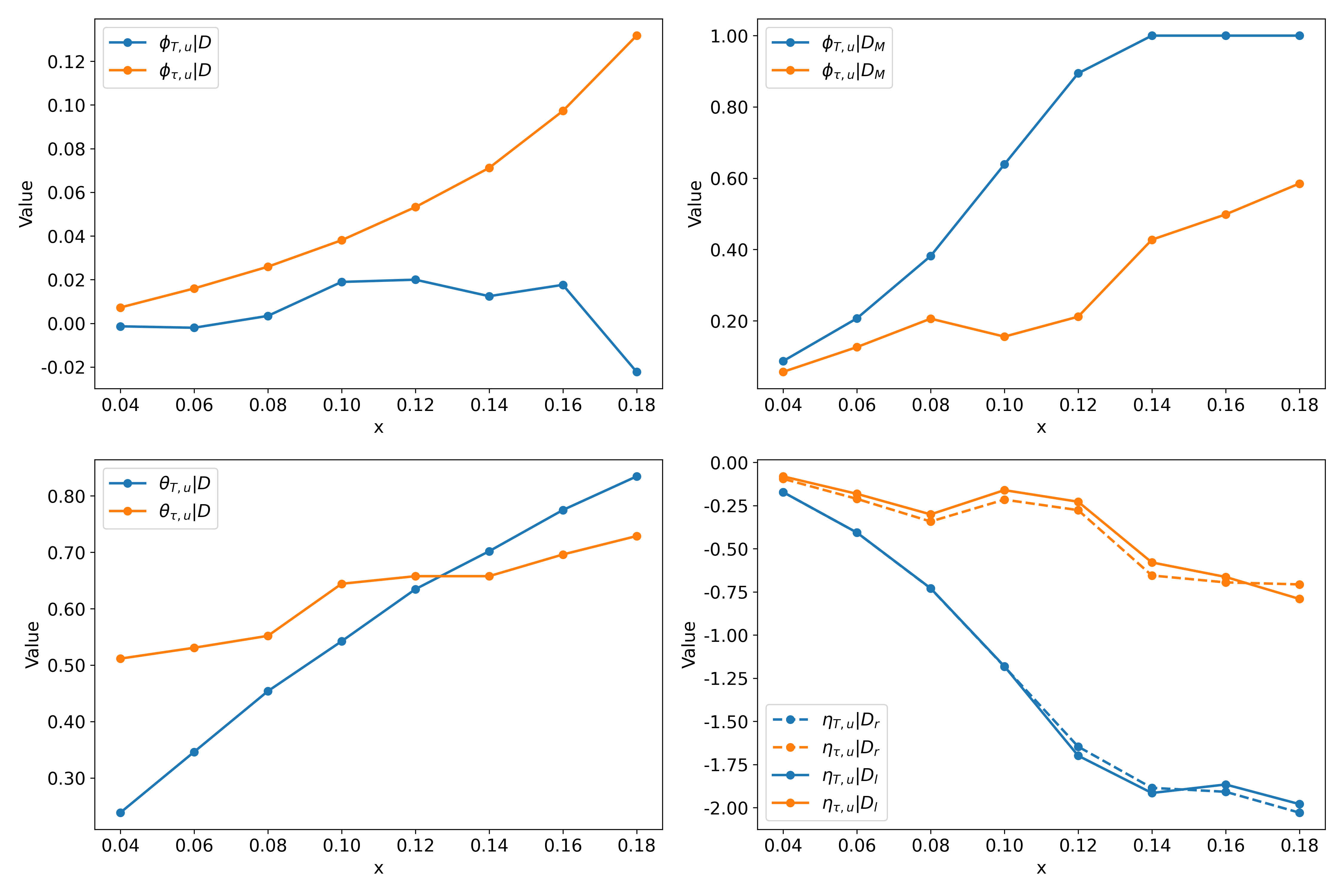}
	\vspace*{-3mm}
	\caption{Influence of $x$ on (a) potential profits $\phi_{T,u}|D$ and $\phi_{\tau,u}|D$; (b) potential profits $\phi_{T,u}|D_M$ and $\phi_{\tau,u}|D_M$; (c) success rates $\theta_{T,u}|D$ and $\theta_{\tau,u}|D$; and (d) risks $\eta_{T,u}|D_r$ and $\eta_{T,u}|D_l$, respectively. }
	\label{Fig:ms}
\end{figure}

Table \ref{tab:ms_metrics} provides comprehensive information about the performance of Iron Condor portfolios under various values of $x$. The metrics $\phi_{\tau,u}|D$, $\theta_{T,u}|D$, $\theta_{\tau,u}|D$, $\phi_{T,u}|D_M$, and $\phi_{\tau,u}|D_M$ exhibit a positive relationship with $x$, while all risk metrics, i.e., $\eta_{T,u}|D_r$, $\eta_{\tau,u}|D_r$, $\eta_{T,u}|D_l$, and $\eta_{\tau,u}|D_l$, show a negative relationship with $x$.  

Notably, the optimal stopping time $\tau$ ranges from 34 to 47 days out of the total 63-day period, corresponding to approximately 54\% to 75\% of the entire duration.  

Interestingly, although employing the optimal stopping strategy reduces profitability in the sideways market ($D_M$), it generally truncates risks and enhances overall profitability on $D$, as evidenced by the $\theta_{\tau,u}|D$ metric.

\begin{table}[h!]
	\centering
        \caption{Portfolio performance metrics for different moneyness}
	\resizebox{\textwidth}{!}{
		\begin{tabular}{lccccccccccc}
			\toprule
			\textbf{$x$} & \textbf{$\phi_{T,u}|D$} & \textbf{$\phi_{\tau,u}|D$} & \textbf{$\tau_{u}|D$} & \textbf{$\theta_{T,u}|D$} & \textbf{$\theta_{\tau,u}|D$} & \textbf{$\phi_{T,u}|D_M$} & \textbf{$\phi_{\tau, u}|D_M$} & \textbf{$\eta_{T,u}|D_r$} & \textbf{$\eta_{\tau,u}|D_r$} & \textbf{$\eta_{T,u}|D_l$} & \textbf{$\eta_{\tau,u}|D_l$} \\
			\midrule
			0.18 & -0.02 & 0.13 & 43 & 0.83 & 0.73 & 1.00 & 0.59 & -2.03 & -0.71 & -1.98 & -0.79 \\
			0.16 & 0.02 & 0.10 & 43 & 0.78 & 0.70 & 1.00 & 0.50 & -1.91 & -0.70 & -1.87 & -0.66 \\
			0.14 & 0.01 & 0.07 & 44 & 0.70 & 0.66 & 1.00 & 0.43 & -1.89 & -0.66 & -1.92 & -0.58 \\
			0.12 & 0.02 & 0.05 & 34 & 0.63 & 0.66 & 0.89 & 0.21 & -1.65 & -0.28 & -1.70 & -0.23 \\
			0.10 & 0.02 & 0.04 & 34 & 0.54 & 0.64 & 0.64 & 0.16 & -1.18 & -0.22 & -1.18 & -0.16 \\
			0.08 & 0.00 & 0.03 & 47 & 0.45 & 0.55 & 0.38 & 0.21 & -0.73 & -0.34 & -0.73 & -0.30 \\
			0.06 & -0.00 & 0.02 & 47 & 0.35 & 0.53 & 0.21 & 0.13 & -0.41 & -0.21 & -0.41 & -0.18 \\
			0.00 & -0.00 & 0.01 & 47 & 0.24 & 0.51 & 0.09 & 0.06 & -0.17 & -0.10 & -0.17 & -0.08 \\
			\bottomrule
		\end{tabular}
	}
	\label{tab:ms_metrics}
\end{table}

\subsection{Influence of Strikes Span}

In this section, we focus on the strikes span \(\hat{x}\), which governs the slope of the call and put spreads in an Iron Condor strategy.  

Figure \ref{Fig:ss_1} presents a narrow-span case with $\hat{x}=0.02$. The risks ($\eta_{t,u}|D_r$ and $\eta_{t,u}|D_l$) are well controlled, as shown in Figures \ref{Fig:ss_1}(c) and (d). However, $\phi_{t,u}|D_M$ is approximately 0, indicating minimal profit potential.  

\begin{figure}[H]
	\centering
	\includegraphics[trim=0cm 0cm 0cm 0cm,clip=true,width=10cm]{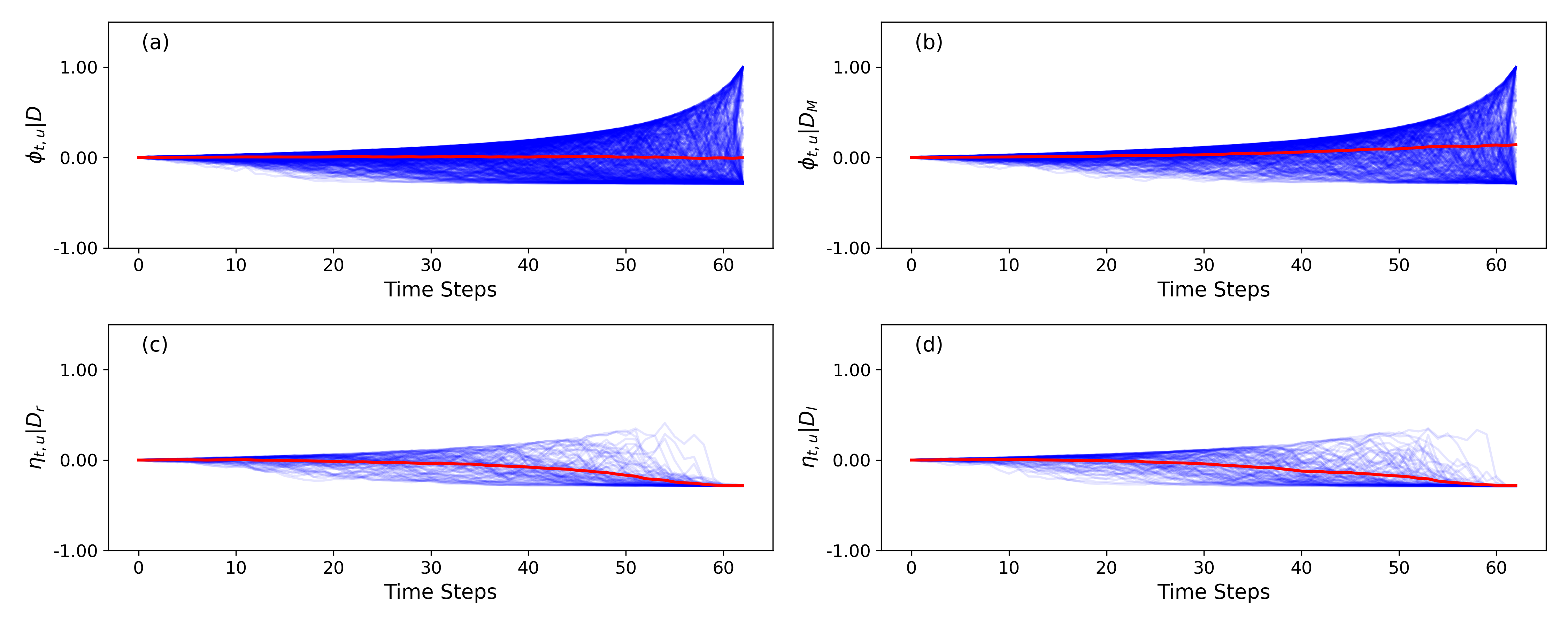}
	\vspace*{-3mm}
	\caption{Distributions of (a) $\phi_{t,u}(\omega)|D$; (b) $\phi_{t,u}(\omega)|D_M$; (c) $\eta_{t,u}|D_r$ and (d) $\eta_{t,u}|D_l$ for the narrow-span portfolio with $\hat{x}=0.02, k=[0.96, 0.98, 1.02, 1.04]$. The red lines represent the expectations.}
	\label{Fig:ss_1}
\end{figure}

In contrast, Figure \ref{Fig:ss_4} illustrates the distribution of $\phi_{t,u}(\omega)$ in a wider-span case with $\hat{x}=0.14$. Here, $\phi_{t,u}|D_M$ exhibits a convex increase (Figure \ref{Fig:ss_4}(b)), but the risks are significantly elevated compared to the $\hat{x}=0.02$ case, as demonstrated in Figures \ref{Fig:ss_4}(c) and (d).

\begin{figure}[H]
	\centering
	\includegraphics[trim=0cm 0cm 0cm 0cm,clip=true,width=10cm]{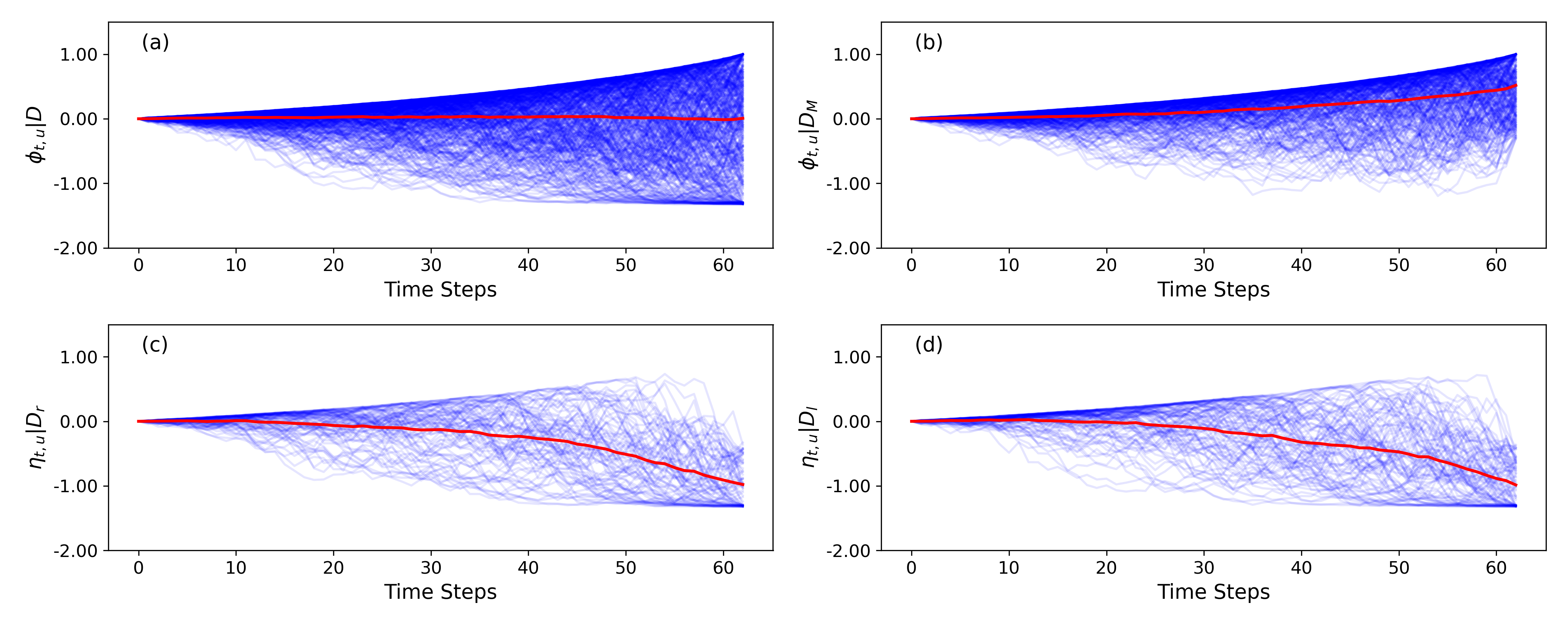}
	\vspace*{-3mm}
	\caption{
    Distributions of (a) $\phi_{t,u}(\omega)|D$; (b) $\phi_{t,u}(\omega)|D_M$; (c) $\eta_{t,u}|D_r$ and (d) $\eta_{t,u}|D_l$ for the wide-span portfolio with  $\hat{x}=0.14, k=[0.84, 0.98, 1.02, 1.16]$. The red lines represent the expectations.}
	\label{Fig:ss_4}
\end{figure}
 
Figure \ref{Fig:ss} illustrates the impact of \(\hat{x}\) on portfolio performance.  In Figure \ref{Fig:ss}(a), $\phi_{\tau,u}|D$ exhibits a linear increase with \(\hat{x}\), highlighting the significance of employing optimal stopping strategies for wide-span portfolios.  

Under sideways market conditions, as shown in Figure \ref{Fig:ss}(b), all optimal stopping strategies determined on $D$ underperform $\phi_{T,u}|D_M$, consistent with the theoretical analysis presented in Section 3.  

Figure \ref{Fig:ss}(c) demonstrates that the success ratio $\theta_{\tau,u}|D$ consistently outperforms $\theta_{T,u}|D$ across all control values of \(\hat{x}\), indicating that optimal stopping strategies enhance the likelihood of portfolio success.  

Furthermore, Figure \ref{Fig:ss}(d) presents the risk metrics in volatile markets. It is evident that for both $D_r$ and $D_l$, the optimal stopping strategies effectively reduce risks.

\begin{figure}[H]
	\centering
	\includegraphics[trim=0cm 0cm 0cm 0cm,clip=true,width=10cm]{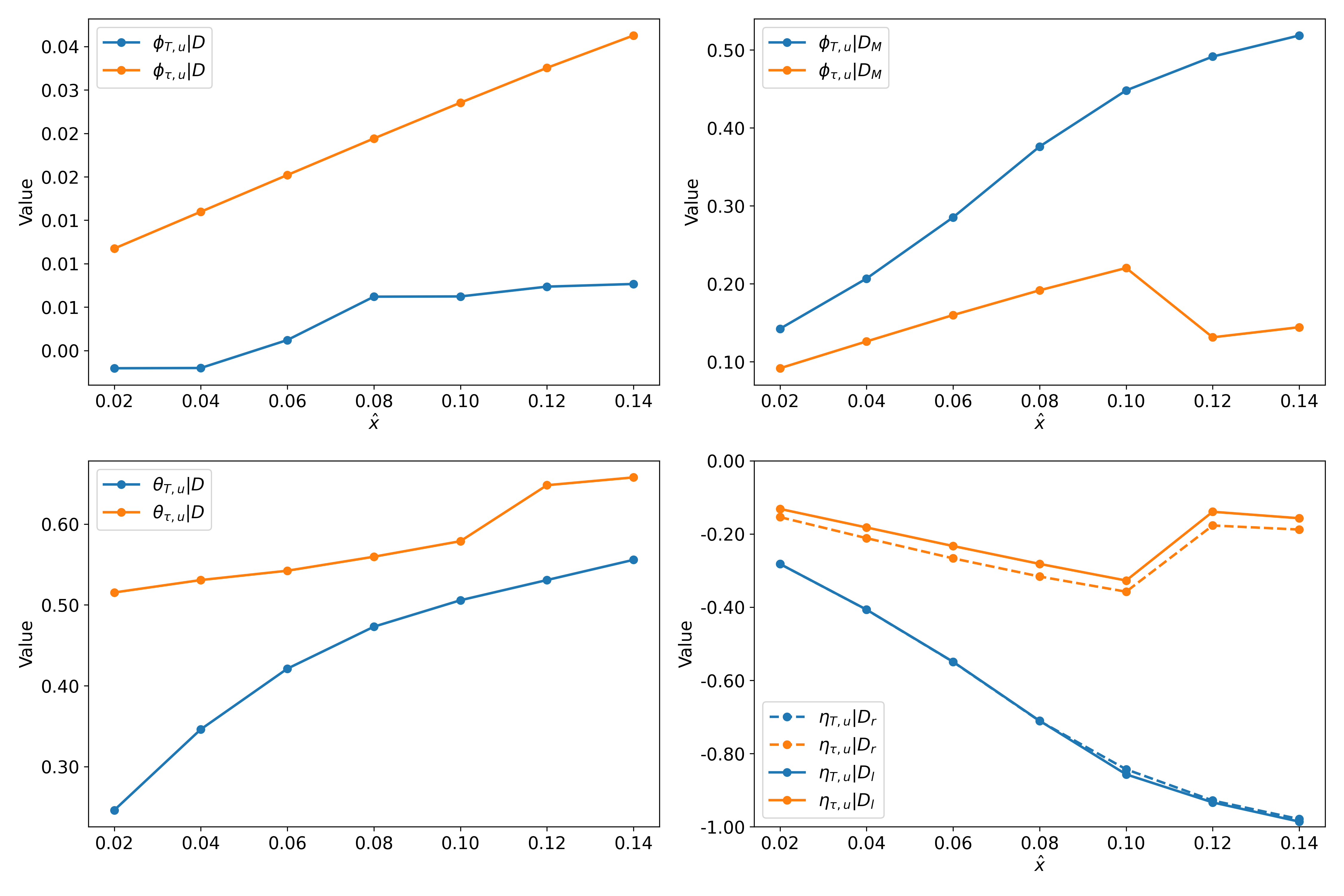}
	\vspace*{-3mm}
	\caption{Influence of $\hat{x}$ on (a) potential profits $\phi_{T,u}|D$ and $\phi_{\tau,u}|D$; (b) potential profits $\phi_{T,u}|D_M$ and $\phi_{\tau,u}|D_M$; (c) success rates $\theta_{T,u}|D$ and $\theta_{\tau,u}|D$; and (d) risks $\eta_{T,u}|D_r$ and $\eta_{T,u}|D_l$, respectively. }
	\label{Fig:ss}
\end{figure}

Table \ref{tab:ss_metrics} summarizes the performance metrics of Iron Condor portfolios for various values of $\hat{x}$.  
Overall, it can be observed that $\phi_{\tau,u}|D$, $\theta_{T,u}|D$, $\theta_{\tau,u}|D$, and $\phi_{T,u}|D_M$ increase as $\hat{x}$ grows, indicating higher potential profits and success rates for wide-span cases. However, $\eta_{T,u}|D_r$ and $\eta_{T,u}|D_l$ also increase with $\hat{x}$, reflecting elevated risks for wider spans.  
The optimal stopping time $\tau$ varies from 34 to 47 days out of the total 63-day period, representing approximately 54\% to 75\% of the entire duration. Implementing optimal stopping strategies slightly improves potential profits while significantly reducing risk magnitudes.

\begin{table}[h!]
	\centering
        \caption{Portfolio performance metrics for different spans}
	\resizebox{\textwidth}{!}{
		\begin{tabular}{lccccccccccc}
			\toprule
			\textbf{$\hat{x}$} & \textbf{$\phi_{T,u}|D$} & \textbf{$\phi_{\tau,u}|D$} & \textbf{$\tau_{u}|D$} & \textbf{$\theta_{T,u}|D$} & \textbf{$\theta_{\tau,u}|D$} & \textbf{$\phi_{T,u}|D_M$} & \textbf{$\phi_{\tau, u}|D_M$} & \textbf{$\eta_{T,u}|D_r$} & \textbf{$\eta_{\tau,u}|D_r$} & \textbf{$\eta_{T,u}|D_l$} & \textbf{$\eta_{\tau,u}|D_l$} \\
			\midrule
		0.02 & 0.00 & 0.00 & 47 &  0.51 &  0.24& 0.14 & 0.09 &  -0.28& -0.15&-0.28&-0.19 \\
		0.04 & 0.00 & 0.01 & 47 & 0.53 &  0.34& 0.20 & 0.12 & -0.40& -0.21& -0.40 &-0.18\\
		0.06 & 0.00 & 0.02 & 47 & 0.54  &  0.42& 0.28 & 0.15 & -0.54&-0.27 &-0.54&-0.23\\
		0.08 & 0.01 & 0.02 & 47  &0.55  & 0.47 & 0.37 & 0.19 & -0.71&-0.32 & -0.71&-0.28 \\
		0.10 & 0.01  &0.03 & 47  & 0.50 &  0.57& 0.44 &  0.22 & -0.84& -0.36&  -0.85&-0.33 \\
		0.12 & 0.01 & 0.03 & 34  & 0.53 &0.64 &0.49 & 0.13 & -0.92&-0.18 & -0.93 &-0.14\\
		0.14 & 0.01 & 0.03 & 34  & 0.55& 0.65 & 0.51 & 0.14&-0.97&-0.19 &-0.98 &-0.16\\
			\bottomrule
		\end{tabular}
	}
	\label{tab:ss_metrics}
\end{table}

\section{Simulation Research on Asymmetry Iron Condor}

This section examines the dynamics of an intriguing Asymmetric Iron Condor portfolio, which exhibits non-martingale behavior based on our simulation results. We define the asymmetry of an Iron Condor portfolio as the unbalanced moneyness between the bull put spread and the bear call spread. When \(\bar{x} > 0\), indicating that the put spread is deeper OTM, the portfolio is referred to as left-biased. Conversely, when \(\bar{x} < 0\), indicating that the call spread is deeper OTM, the portfolio is called right-biased.  

To clarify the analysis, we first present a balanced baseline in Figure \ref{Fig:hs_3} and then provide a comparative discussion of the left-biased and right-biased portfolios in Figures \ref{Fig:hs_0} and \ref{Fig:hs_6}

\begin{figure}[H]
	\centering
	\includegraphics[trim=0cm 0cm 0cm 0cm,clip=true,width=10cm]{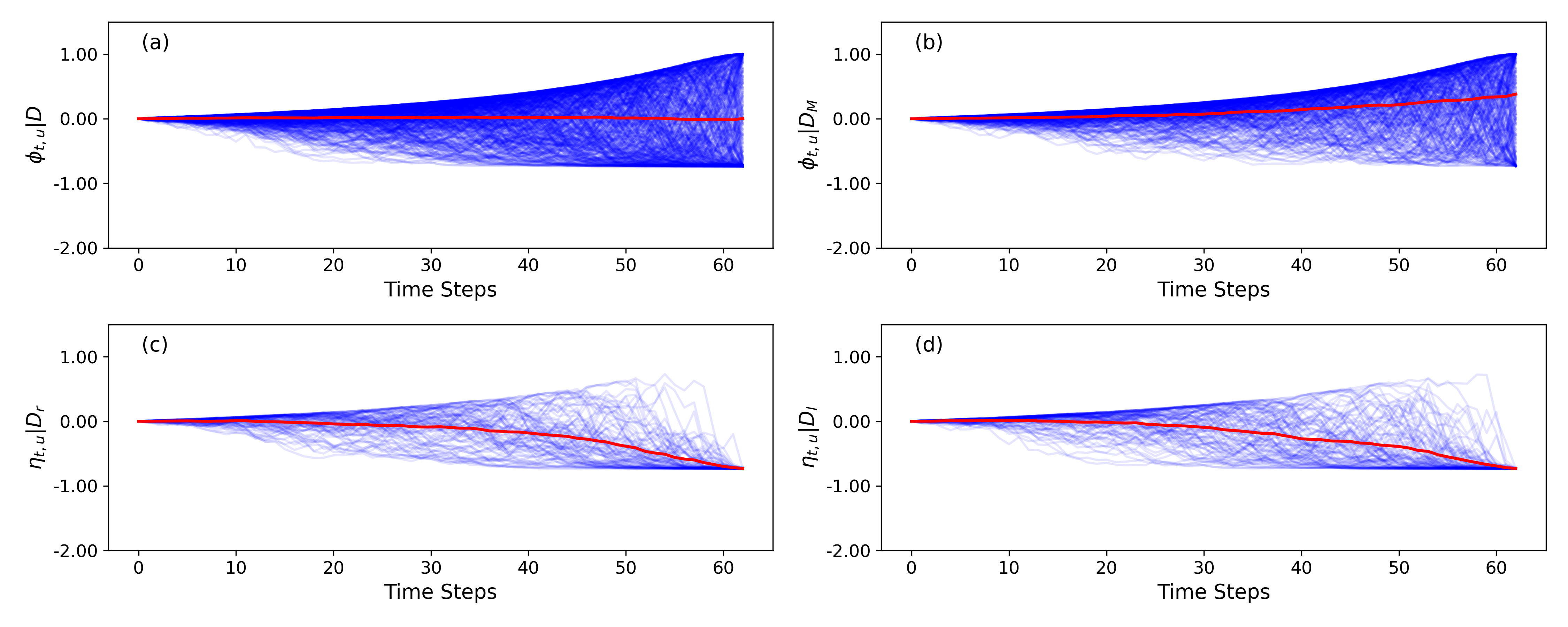}
	\vspace*{-3mm}
	\caption{Distributions of (a) $\phi_{t,u}(\omega)|D$; (b) $\phi_{t,u}(\omega)|D_M$; (c) $\eta_{t,u}|D_r$ and (d) $\eta_{t,u}|D_l$ for the symmetric portfolio with  $\bar{x}=0, k=[0.92, 0.96, 1.04, 1.08]$. The red lines represent the expectations.}
	\label{Fig:hs_3}
\end{figure}

As shown in Figure \ref{Fig:hs_0}(a), $\phi_{t,u}|D$ exhibits a slightly concave shape, suggesting non-martingale properties that may indicate arbitrage opportunities. Furthermore, $\eta_{t,u}|D_r$ and $\eta_{t,u}|D_l$ become asymmetric: the risk of $\eta_{t,u}|D_r$ increases significantly compared to the symmetric case shown in Figure \ref{Fig:hs_3}(c), while the values of $\eta_{t,u}|D_l$ also increase compared to Figure \ref{Fig:hs_3}(d). Notably, $\eta_{t,u}|D_l$ achieves positive values within $t \in [0, 45]$, indicating reduced risk during this time interval.  

\begin{figure}[H]
	\centering
	\includegraphics[trim=0cm 0cm 0cm 0cm,clip=true,width=10cm]{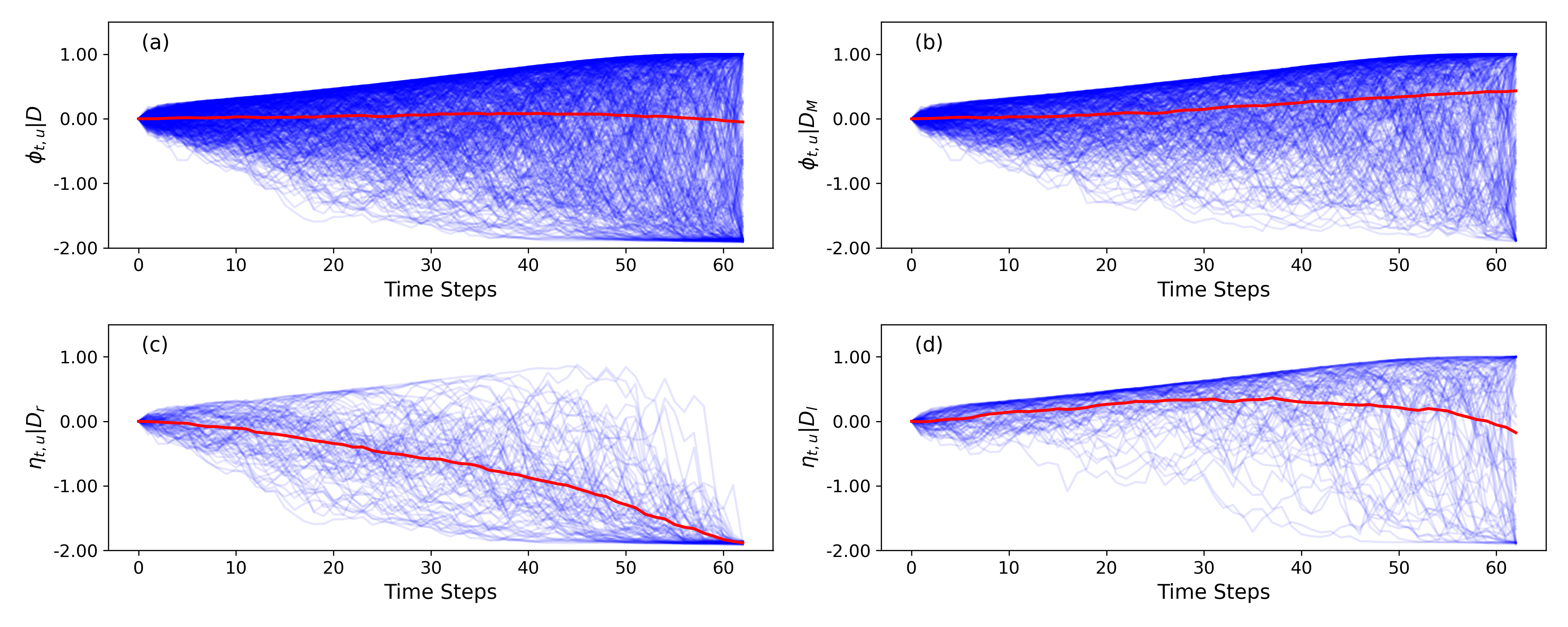}
	\vspace*{-3mm}
	\caption{Distributions of  (a) $\phi_{t,u}(\omega)|D$; (b) $\phi_{t,u}(\omega)|D_M$; (c) $\eta_{t,u}|D_r$ and (d) $\eta_{t,u}|D_l$ for the left-based portfolio with $\bar{x}=0.10, k=[0.82, 0.86, 1.04, 1.08]$. The red lines represent the expectations.}
	\label{Fig:hs_0}
\end{figure}

Figure \ref{Fig:hs_6} illustrates the right-biased case. Comparing Figure \ref{Fig:hs_6}(a) with Figure \ref{Fig:hs_3}(a), the Asymmetric Iron Condor introduces additional risk. Specifically, $\phi_{t,u}|D_M$ exhibits a convex increase. The risk metric $\eta_{t,u}|D_r$, shown in Figure \ref{Fig:hs_6}(c), displays a concave shape and achieves positive values before the 50th day, which is similar to the behavior of $\eta_{t,u}|D_l$ observed in Figure \ref{Fig:hs_0}(d).  
However, in Figure \ref{Fig:hs_6}(d), $\eta_{t,u}|D_l$ expands compared to that in Figure \ref{Fig:hs_3}(d), indicating increased risk.  

\begin{figure}[H]
	\centering
	\includegraphics[trim=0cm 0cm 0cm 0cm,clip=true,width=10cm]{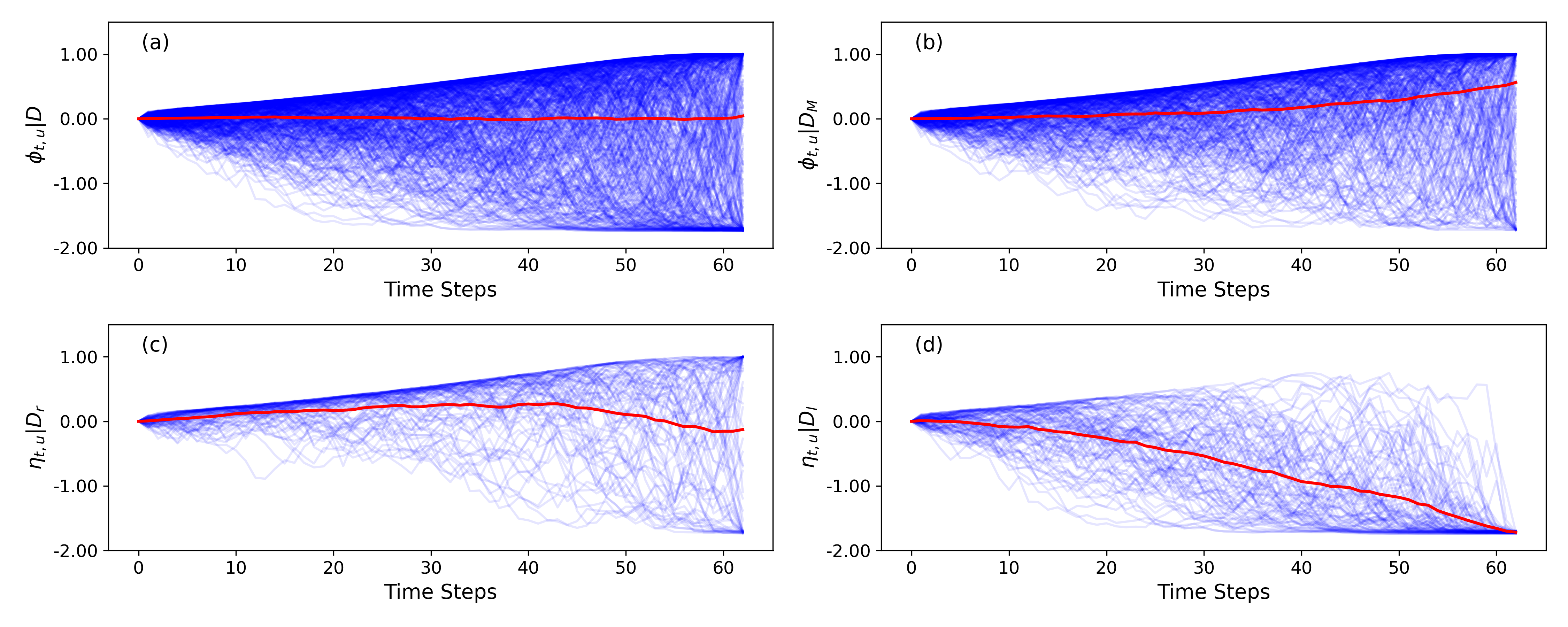}
	\vspace*{-3mm}
	\caption{Distributions of  (a) $\phi_{t,u}(\omega)|D$; (b) $\phi_{t,u}(\omega)|D_M$; (c) $\eta_{t,u}|D_r$ and (d) $\eta_{t,u}|D_l$ for the right-based portfolio with $\bar{x}=-0.10, k=[0.92, 0.96, 1.14, 1.18]$. The red lines represent the expectations.}
	\label{Fig:hs_6}
\end{figure}

Figure \ref{Fig:hs} illustrates the impact of asymmetry on portfolio performance. In Figure \ref{Fig:hs}(a), for left-biased portfolios, $\phi_{T,u}|D$ and $\phi_{\tau,u}|D$ overlap, indicating $\tau = T$. Moreover, the large values of $\phi_{T,u}|D$ occurs in $x<0$ suggests that left-biased portfolios are more profitable.  
In contrast, for right-biased portfolios, $\phi_{T,u}|D$ underperforms the symmetric case, while $\phi_{\tau,u}|D$ surpasses both the symmetric case and the left-biased portfolios. This highlights the importance of employing an optimal stopping strategy in right-biased portfolios.  

Under sideways market conditions, as shown in Figure \ref{Fig:hs}(b), the left-biased portfolio significantly outperforms other scenarios. Furthermore, $\theta_{T,u}|D$ and $\theta_{\tau,u}|D$ achieve identical values, indicating that the optimal stopping time for the left-biased portfolio under sideways market conditions is $T$.

Despite the data generator using tuned parameters from the real SPX market, some imperfect symmetric patterns can still be observed. In Figure \ref{Fig:hs}(c), both $\theta_{T,u}|D$ and $\theta_{\tau,u}|D$ exhibit even symmetry around the symmetric portfolio. This indicates that asymmetric portfolios can achieve higher success rates compared to symmetric portfolios in the SPX market, although the latter case is more complex to analyze.

Moreover, in Figure \ref{Fig:hs}(d), $\eta_{T,u}|D_r$ and $\eta_{T,u}|D_l$ display odd symmetry about the symmetric portfolio. For the left-biased portfolio, $\eta_{\tau,u}|D_l$ outperforms $\eta_{T,u}|D_l$, while $\eta_{\tau,u}|D_r$ underperforms $\eta_{T,u}|D_r$. This implies that when trading a left-biased portfolio, one should adopt an optimal stopping strategy in bearish markets but refrain from using such a strategy in bullish markets. Conversely, for the right-biased portfolio, the strategy should be reversed.

\begin{figure}[H]
	\centering
	\includegraphics[trim=0cm 0cm 0cm 0cm,clip=true,width=10cm]{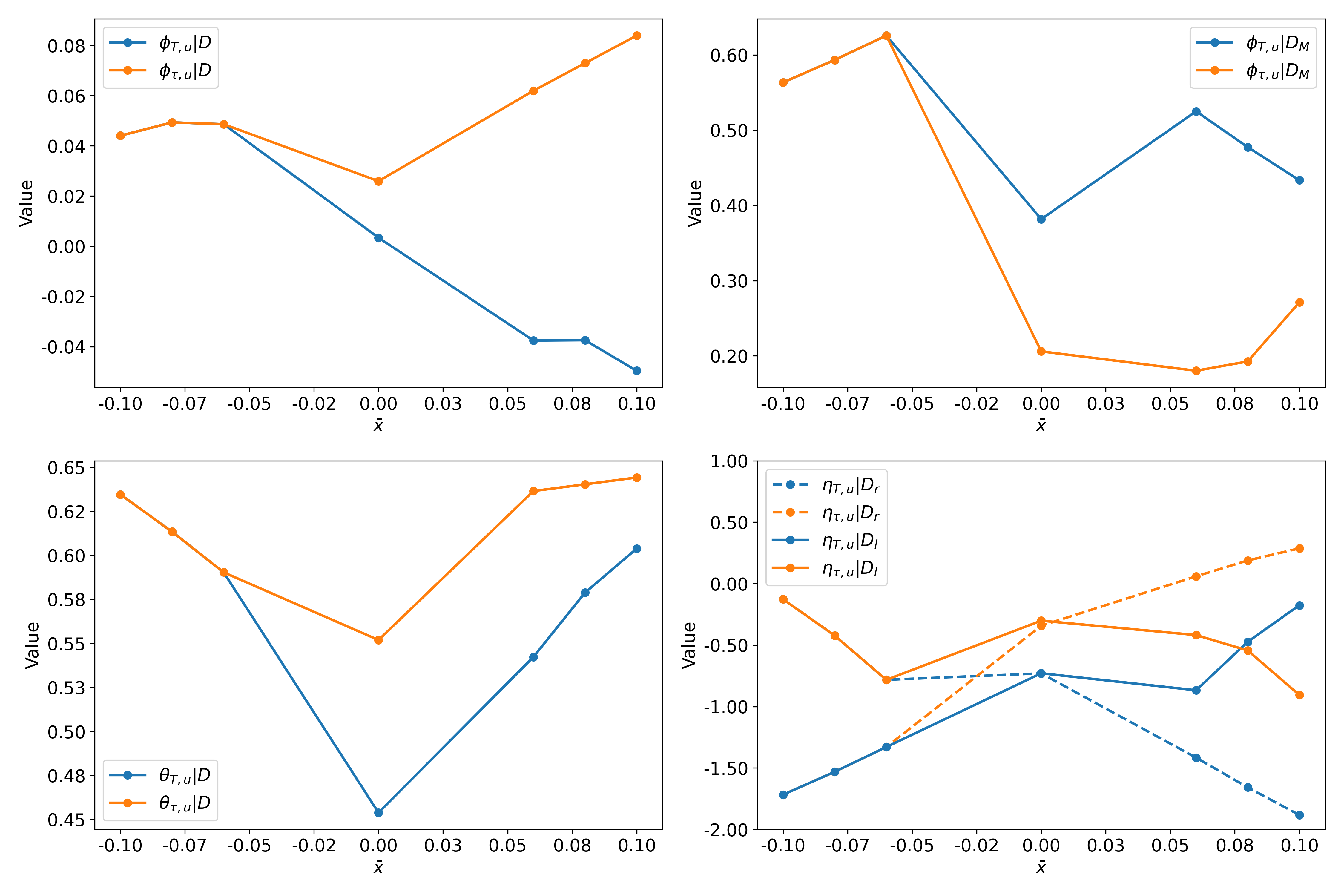}
	\vspace*{-3mm}
	\caption{Influence of $\bar{x}$ on (a) potential profits $\phi_{T,u}|D$ and $\phi_{\tau,u}|D$; (b) potential profits $\phi_{T,u}|D_M$ and $\phi_{\tau,u}|D_M$; (c) success rates $\theta_{T,u}|D$ and $\theta_{\tau,u}|D$; (d) risks $\eta_{T,u}|D_r$, $\eta_{\tau,u}|D_r$, $\eta_{T,u}|D_l$ (dash), and $\eta_{\tau,u}|D_l$ (dash), respectively. }
	\label{Fig:hs}
\end{figure}

Table \ref{tab:hs_metrics} summarizes the influence of \(\bar{x}\) on portfolio performance. Overall, the left-biased portfolio is observed to outperform the right-biased portfolio for the following reasons:  
1. \(\phi_{T,u}|D > 0\) for all \(\bar{x} < 0\), whereas \(\phi_{T,u}|D < 0\) when \(\bar{x} > 0\).  
2. \(\phi_{\tau,u}|D_M\) of the left-biased portfolio strictly dominates that of the right-biased portfolio.  
3. Under the same absolute level of \(|\bar{x}|\), \(\theta_{T,u}|D\) of the left-biased portfolio conditionally dominates that of the right-biased portfolio.  
4. The maximum risk of the left-biased portfolio is lower than that of the right-biased ones.  

However, the right-biased Iron Condor heavily depends on precise optimal stopping time to truncate extreme losses before expiration, which may cause \( \phi_{\tau,u}|D, x > 0 \) to be dominated by \( \phi_{T,u}|D, x \leq 0 \). Nevertheless, it leads to \( \eta_{\tau,u}|D_r, x > 0 \) strictly dominating \( \eta_{\tau,u}|D_l, x \leq 0 \), and \( \eta_{\tau,u}|D_l, x > 0 \) strictly dominating \( \eta_{\tau,u}|D_r, x \leq 0 \).

The results are intriguing, and we aim to analyze them further.  
Since the S\&P 500 index is predominantly bullish over time, establishing a left-biased portfolio provides greater tolerance for the upward movement of the underlying price while ensuring the terminal price remains within the profitable region. As a result, the left-biased Iron Condor is appealing for long-term holding and consistently outperforms the commonly used symmetric Iron Condor strategy.  

From Table \ref{tab:hs_metrics}, the optimal stopping time \(\tau\) for the right-biased portfolio varies between 34 and 41 days out of the total 63-day period, representing approximately 54\% to 65\% of the entire duration.

\begin{table}[h!]
	\centering
        \caption{Portfolio performance metrics for different asymmetry degrees}
	\resizebox{\textwidth}{!}{
		\begin{tabular}{lccccccccccc}
			\toprule
			\textbf{$\bar{x}$} & \textbf{$\phi_{T,u}|D$} & \textbf{$\phi_{\tau,u}|D$} & \textbf{$\tau_{u}|D$} & \textbf{$\theta_{T,u}|D$} & \textbf{$\theta_{\tau,u}|D$} & \textbf{$\phi_{T,u}|D_M$} & \textbf{$\phi_{\tau, u}|D_M$} & \textbf{$\eta_{T,u}|D_r$} & \textbf{$\eta_{\tau,u}|D_r$} & \textbf{$\eta_{T,u}|D_l$} & \textbf{$\eta_{\tau,u}|D_l$} \\
			\midrule
		0.10 & -0.05 & 0.08 & 41 & 0.60 & 0.64 & 0.43 & 0.27 & -1.88  & 0.29& -0.18&-0.91 \\
		0.08 &-0.04 & 0.07 & 34 & 0.58 & 0.64 & 0.48 & 0.19 & -1.66 &0.19 & -0.47&-0.54 \\
		0.06 &-0.04 & 0.06 & 34 & 0.54 & 0.64 & 0.52 & 0.18 & -1.42 & 0.06 & -0.87& -0.42\\
		0.00 &0.00 & 0.03 & 47 & 0.45 & 0.55 & 0.38 & 0.21 & -0.73 &-0.34 & -0.73& -0.3 \\
		-0.06 &0.05 & 0.05 & 62 & 0.59 & 0.59 & 0.63 & 0.63 & -0.78 &-1.33 & -1.33& -0.78 \\
		-0.08 &0.05 & 0.05 & 62 & 0.61 & 0.61 & 0.59 & 0.59 & -0.42 &-1.53& -1.53&-0.42 \\
		-0.10 &0.04 & 0.04 & 62 & 0.63 & 0.63 & 0.56 & 0.56 & -0.13 &-1.72& -1.72& -0.13\\
			\bottomrule
		\end{tabular}
	}
		\label{tab:hs_metrics}
		\end{table}

\section{Validation on Actual SPX Market}

This section examines the performance of Iron Condor portfolios in the actual SPX market to validate the prior findings. We replace \(\phi_{t,u}\) with (normalized) Profit \& Loss (P\&L) as the metric since only a single realization is available in real world.  Specifically, we analyze three cases, with time frames spanning 63, 63, and 49 trading days prior to the options' maturity dates of 2020-12-18, 2021-07-16, and 2022-10-21, respectively. The data used in this analysis is sourced from the OptionMetrics database.  

Figure \ref{Fig:real_bullish} presents the performance of portfolios with varying \(x\), \(\hat{x}\), and \(\bar{x}\) during a bullish market (Figure \ref{Fig:real_bullish}(a)), where \(S_T > 3700\).

Figure \ref{Fig:real_bullish}(b) shows the P\&L of portfolios with different \(x\) values. Only the portfolio with the maximum \(x\) (purple line) achieves its maximum profit, while 3 out of 5 portfolios result in negative returns, indicating a low success rate in this scenario.  

Figure \ref{Fig:real_bullish}(c) shows the influence of \(\hat{x}\). Unfortunately, none of these portfolios achieve positive return. 

Promising results emerge in Figure \ref{Fig:real_bullish}(d), where a mirror symmetry of left-biased portfolios and right-biased portfolios about the symmetric portfolio is observed, which is very interesting. All left-biased portfolios achieve their maximum profit, whereas symmetric and right-biased portfolios all get negative return. This align with our simulation results that the left-biased Iron Condor can effectively handle risks resulting from bullish markets. 

\begin{figure}[H]
	\centering
	\includegraphics[trim=0cm 0cm 0cm 0cm,clip=true,width=10cm]{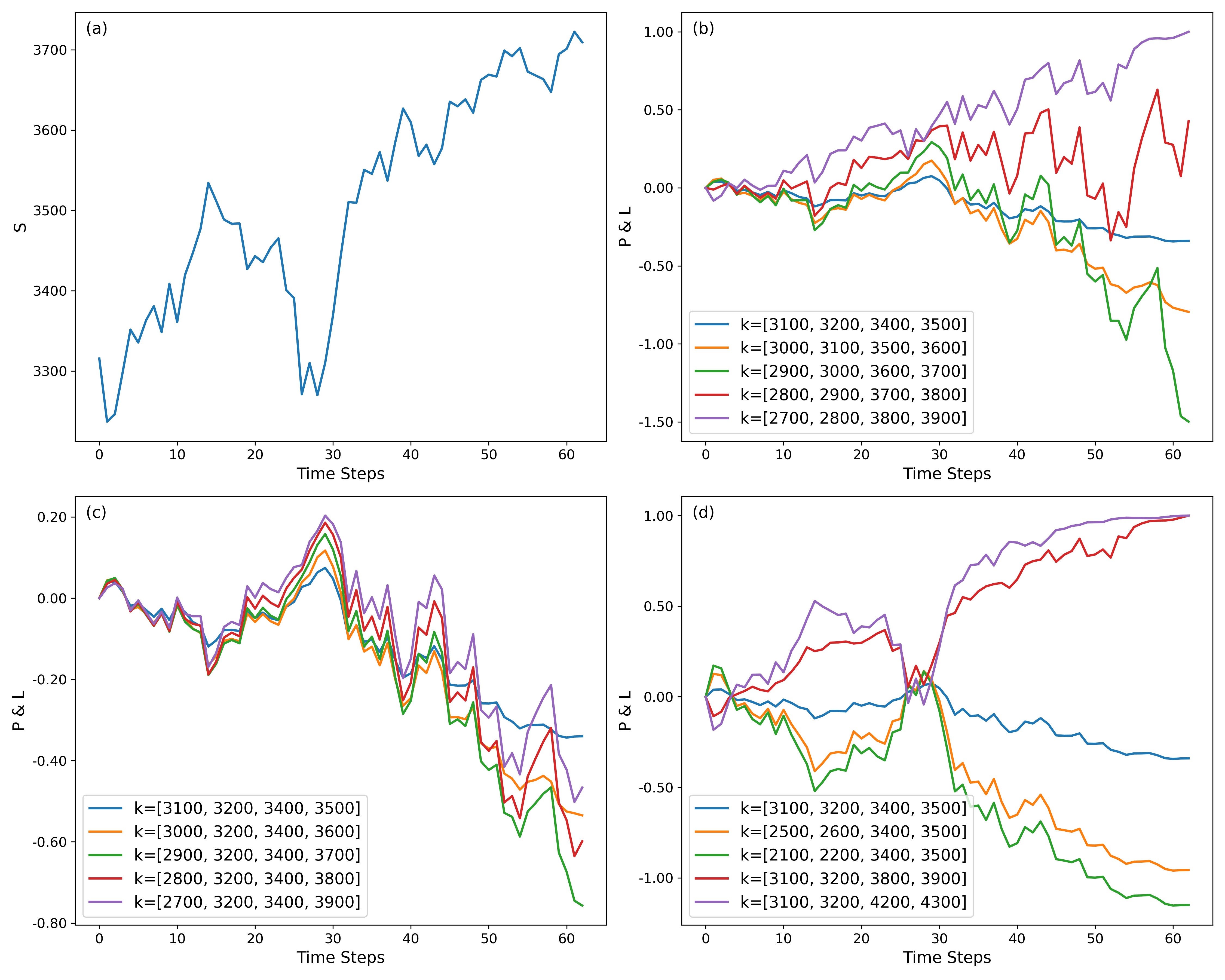}
	\vspace*{-3mm}
        \caption{ (a) Trajectory of underlying prices $S_t$; and portfolios P\&L influenced by (b) moneyness $x$; (c) strike span $\hat{x}$; and (d) asymmtry degree $\bar{x}$ values, during a bullish market observed over 63 trading days prior to the options' maturity date of 2020-12-18.}
	\label{Fig:real_bullish}
\end{figure}

Figure \ref{Fig:real_sideway} illustrates the P\&L performance of portfolios under a sideways market condition (Figure \ref{Fig:real_sideway}(a)).
In Figure \ref{Fig:real_sideway}(b), four out of five portfolios achieve their maximum profits, except the one with the smallest $x$. This result aligns with our simulation findings that increasing $x$ improves the success rate.
In Figure \ref{Fig:real_sideway}(c), for portfolios with a fixed $x$, varying $\hat{x}$ results in similar P\&L patterns. Notably, three out of five portfolios yield positive returns.
In Figure \ref{Fig:real_sideway}(d), all left-biased portfolios achieve their maximum profits, whereas the symmetric and right-biased portfolios exhibit poorer performance. This observation further supports the outcomes from our simulations.
Moreover, adopting an optimal stopping strategy within the time interval [40, 50] days can significantly reduce overall risks, even if the profits of some portfolios are slightly diminished.

\begin{figure}[H]
	\centering
	\includegraphics[trim=0cm 0cm 0cm 0cm,clip=true,width=10cm]{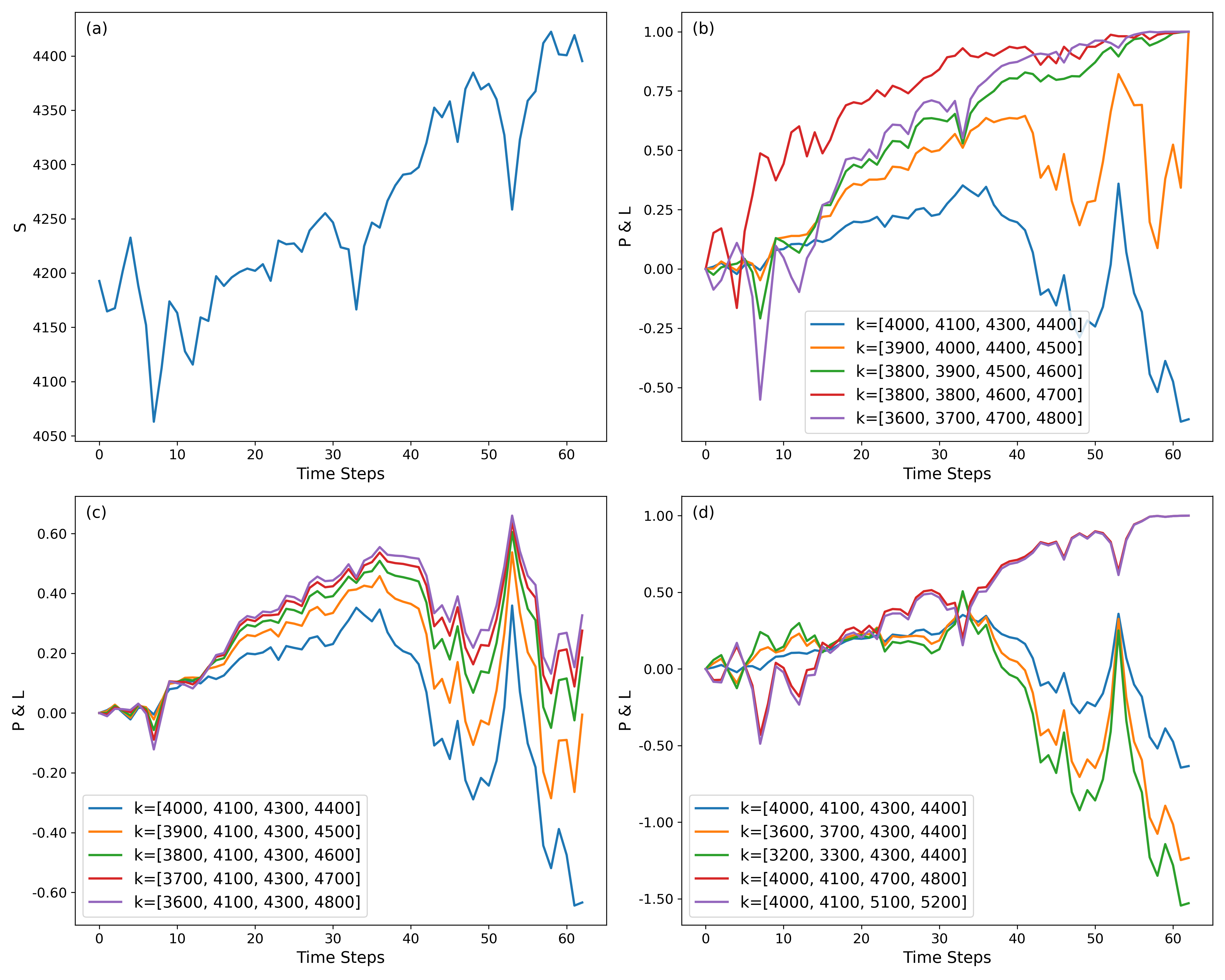}
	\vspace*{-3mm}
	\caption{(a) Trajectory of underlying prices $S_t$; (b) portfolios P\&L under varying $x$ values; (c) portfolios P\&L under varying $\hat{x}$ values; and (d) portfolios P\&L under varying $\bar{x}$ values, during a sideway market observed over 63 trading days prior to the options' maturity date of 2020-07-16}
\label{Fig:real_sideway}
\end{figure}

Finally, Figure \ref{Fig:real_bearish}(a) depicts an extreme bearish market. Under such conditions, changes in either $x$ or $\hat{x}$ result in losses, as shown in Figures \ref{Fig:real_bearish}(b) and \ref{Fig:real_bearish}(c). 
However, constructing an asymmetric right-biased Iron Condor portfolio leads to positive returns, aligning with our findings from simulations.

\begin{figure}[H]
	\centering
	\includegraphics[trim=0cm 0cm 0cm 0cm,clip=true,width=10cm]{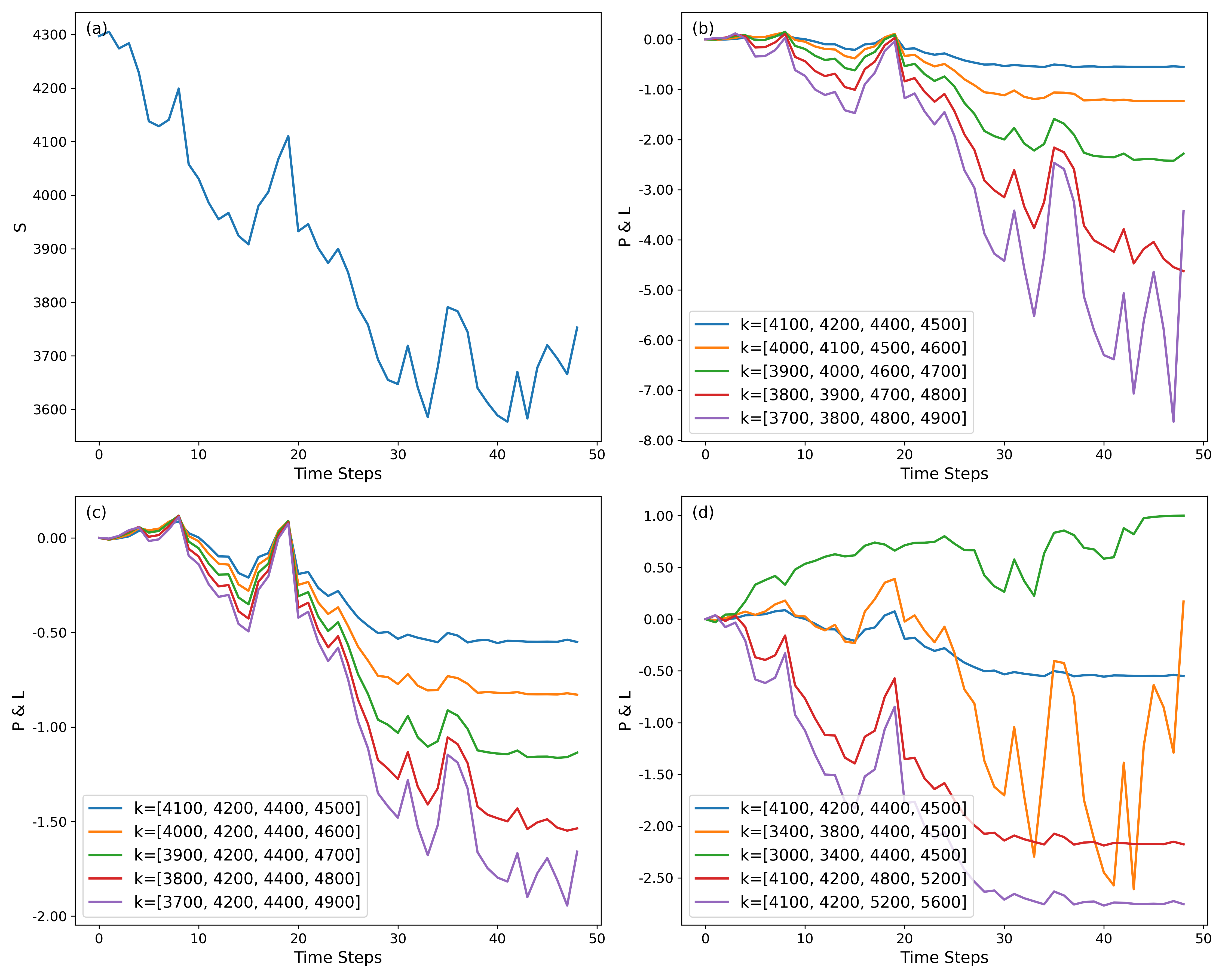}
	\vspace*{-3mm}
	\caption{(a) Trajectory of underlying prices $S_t$; (b) portfolios P\&L under varying $x$ values; (c) portfolios P\&L under varying $\hat{x}$ values; and (d) portfolios P\&L under varying $\bar{x}$ values, during a bearish market observed over 49 trading days prior to the options' maturity date of 2022-10-21}
	\label{Fig:real_bearish}
\end{figure}

\section{Conclusion}

This paper provides an in-depth analysis of the influence of the control process \( u(k_i,\tau) \) (\( u(x,\hat{x},\bar{x},\tau) \)) on Iron Condor portfolios.

Initially, we assume the underlying price process \( S_t \) is a bounded martingale within \([K_2, K_3]\), and then prove a theorem that the optimal stopping time is \( \tau = T \) if the Iron Condor strike structure satisfies \( k_1 < k_2 = K_2 < S_t < k_3 = K_3 < k_4 \) for all \( t \in [0, T] \).

We extend our study to more general scenarios by employing a simulation method based on the data generator derived from the Rough Heston model. We parameterize  \( u \) using four variables: moneyness (\(x\)), strike span (\(\hat{x}\)), asymmetry degree (\(\bar{x}\)) and optimal stopping time \(\tau\). The simulation results reveal the following key findings:
(1) Asymmetric, left-biased Iron Condor portfolios with \( \tau = T \) tend to be optimal in SPX markets, balancing profitability and risk management effectively; 
(2) Deep OTM Iron Condor portfolios improve profitability and success rates, but they also introduce the risk of extreme losses. Adopting an optimal stopping strategy can mitigate these losses, although it slightly reduces potential profits;
(3) \(\tau\) generally falls between 50\% and 75\% of the total duration, except for left-biased portfolios.

Finally, we validate our findings on the actual SPX market through three case studies covering bullish, sideways, and bearish market conditions, and the results support the simulation findings. 

There are two limitations to the current research: 
\begin{enumerate}
    \item The optimal stopping strategies derived from simulation methods need to be quantified and further analyzed.
    \item The approach requires a market trend forecasting method for effective strategy design.
\end{enumerate}

The first limitation may be addressed using entropy-based methods, while the second could be further investigated through machine learning techniques in future research.

\section*{Declarations}

\begin{itemize}
	\item Funding:\\
	      This work is supported by National Natural Fundation of China (No.42402239) and Jiangsu University of Science and Technology Scientific Research Start-up Fund (No.1132932306).
	\item Competing Interest declaration:\\
	      The author declare that they have no known competing financial interests or personal relationships that could have appeared to influence the work reported in this paper.
	\item Data Availability Statement:\\
	      The author do not have permission to share data.
	\item Author Contributions:\\
	      Qiguo Sun: Conceptualization, Investigation, Methodology, Analysis, Writing - Review\& Editing. \\
            Hanyue Huang: Conceptualization, Investigation, Methodology, Analysis, Writing - Review\& Editing. \\
           Xibei Yang:           Methodology
          
	\item Ethical statement:\\
	      All data handling procedures adhered to ethical guidelines for research.

\end{itemize}

\bibliographystyle{plainnat}
\bibliography{latex_files/mybib}

\begin{thebibliography}{16}
\providecommand{\natexlab}[1]{#1}
\providecommand{\url}[1]{\texttt{#1}}
\expandafter\ifx\csname urlstyle\endcsname\relax
  \providecommand{\doi}[1]{doi: #1}\else
  \providecommand{\doi}{doi: \begingroup \urlstyle{rm}\Url}\fi

\bibitem[Abi~Jaber(2019)]{abi2019lifting}
Eduardo Abi~Jaber.
\newblock Lifting the heston model.
\newblock \emph{Quantitative finance}, 19\penalty0 (12):\penalty0 1995--2013, 2019.

\bibitem[Bayer et~al.(2016)Bayer, Friz, and Gatheral]{bayer2016pricing}
Christian Bayer, Peter Friz, and Jim Gatheral.
\newblock Pricing under rough volatility.
\newblock \emph{Quantitative Finance}, 16\penalty0 (6):\penalty0 887--904, 2016.

\bibitem[Bennedsen et~al.(2017)Bennedsen, Lunde, and Pakkanen]{bennedsen2017hybrid}
Mikkel Bennedsen, Asger Lunde, and Mikko~S Pakkanen.
\newblock Hybrid scheme for brownian semistationary processes.
\newblock \emph{Finance and Stochastics}, 21:\penalty0 931--965, 2017.

\bibitem[Cohen(2005)]{cohen2005bible}
Guy Cohen.
\newblock \emph{The bible of options strategies: the definitive guide for practical trading strategies}.
\newblock Pearson Education, 2005.

\bibitem[de~Saint-Cyr(2023)]{de2023simple}
Alberic de~Saint-Cyr.
\newblock A simple historical analysis of the performance of iron condors on the spx.
\newblock \emph{Available at SSRN 4643378}, 2023.

\bibitem[Dupire et~al.(1994)]{dupire1994pricing}
Bruno Dupire et~al.
\newblock Pricing with a smile.
\newblock \emph{Risk}, 7\penalty0 (1):\penalty0 18--20, 1994.

\bibitem[Dziawgo(2020)]{dziawgo2020iron}
Ewa Dziawgo.
\newblock The iron condor strategy in financial risk management.
\newblock \emph{Prace Naukowe Uniwersytetu Ekonomicznego we Wroc{\l}awiu}, 64\penalty0 (2):\penalty0 33--44, 2020.

\bibitem[El~Euch et~al.(2019)El~Euch, Gatheral, and Rosenbaum]{el2019roughening}
Omar El~Euch, Jim Gatheral, and Mathieu Rosenbaum.
\newblock Roughening heston.
\newblock \emph{Risk}, pages 84--89, 2019.

\bibitem[Fadugba(2020)]{fadugba2020homotopy}
Sunday~Emmanuel Fadugba.
\newblock Homotopy analysis method and its applications in the valuation of european call options with time-fractional black-scholes equation.
\newblock \emph{Chaos, Solitons \& Fractals}, 141:\penalty0 110351, 2020.

\bibitem[Gatheral and Jacquier(2014)]{gatheral2014arbitrage}
Jim Gatheral and Antoine Jacquier.
\newblock Arbitrage-free svi volatility surfaces.
\newblock \emph{Quantitative Finance}, 14\penalty0 (1):\penalty0 59--71, 2014.

\bibitem[Heston(1993)]{heston1993closed}
Steven~L Heston.
\newblock A closed-form solution for options with stochastic volatility with applications to bond and currency options.
\newblock \emph{The review of financial studies}, 6\penalty0 (2):\penalty0 327--343, 1993.

\bibitem[Hu and {\O}ksendal(2003)]{hu2003fractional}
Yaozhong Hu and Bernt {\O}ksendal.
\newblock Fractional white noise calculus and applications to finance.
\newblock \emph{Infinite dimensional analysis, quantum probability and related topics}, 6\penalty0 (01):\penalty0 1--32, 2003.

\bibitem[Ma and Wu(2022)]{ma2022fast}
Jingtang Ma and Haofei Wu.
\newblock A fast algorithm for simulation of rough volatility models.
\newblock \emph{Quantitative Finance}, 22\penalty0 (3):\penalty0 447--462, 2022.

\bibitem[Wang et~al.(2022)Wang, Wen, Yang, and Shao]{wang2022practical}
Jian Wang, Shuai Wen, Mengdie Yang, and Wei Shao.
\newblock Practical finite difference method for solving multi-dimensional black-scholes model in fractal market.
\newblock \emph{Chaos, Solitons \& Fractals}, 157:\penalty0 111895, 2022.

\bibitem[Wong and Bilokon(2024)]{wong2024simulation}
Yat Chun~Chester Wong and Paul Bilokon.
\newblock Simulation of fractional brownian motion and related stochastic processes in practice: A straightforward approach.
\newblock \emph{Available at SSRN}, 2024.

\bibitem[Woodard(2011)]{woodard2011iron}
Jared Woodard.
\newblock \emph{Iron Condor Spread Strategies: Timing, Structuring, and Managing Profitable Options Trades}.
\newblock Pearson Education, 2011.

\end{thebibliography}

\end{document}